\newcommand{\CLASSINPUTtoptextmargin}{0.78 in}
\newcommand{\CLASSINPUTbottomtextmargin}{1.05in}

\documentclass[conference]{IEEEtran}
\usepackage{cite}
\usepackage{amsmath,amssymb,amsfonts,amsthm,dsfont,bm}
\usepackage{algorithmic}
\usepackage{graphicx}
\usepackage{textcomp}
\usepackage{xcolor}
\usepackage{hyperref}

\newtheorem{theorem}{Theorem}
\newtheorem{lemma}{Lemma}

\newtheorem{corollary}{Corollary}
\newtheorem{definition}{Definition}

\newtheorem{notation}{Notation}
\newtheorem{remark}{Remark}
\newtheorem{annotation}{Annotation}

\newcommand{\nn}{\nonumber}

\def\mI{{\bm{I}}}

\def\BibTeX{{\rm B\kern-.05em{\sc i\kern-.025em b}\kern-.08em
    T\kern-.1667em\lower.7ex\hbox{E}\kern-.125emX}}
\begin{document}
\title{
Information-theoretic Analysis of the Gibbs Algorithm: An Individual Sample Approach
}

\author{\IEEEauthorblockN{Youheng Zhu}
\IEEEauthorblockA{\textit{School of Computer Science and Technology} \\
\textit{Huazhong University of Science and Technology}\\
Wuhan, China \\
email: youhengzhu@hust.edu.cn}
\and
\IEEEauthorblockN{Yuheng Bu}
\IEEEauthorblockA{\textit{ECE department} \\
\textit{University of Florida}\\
Gainesville, USA \\
email: buyuheng@ufl.edu}
}

\maketitle

\begin{abstract}
Recent progress has shown that the generalization error of the Gibbs algorithm can be exactly characterized using the symmetrized KL information between the learned hypothesis and the entire training dataset. However, evaluating such a characterization is cumbersome, as it involves a high-dimensional information measure. In this paper, we address this issue by considering individual sample information measures within the Gibbs algorithm. Our main contribution lies in establishing the asymptotic equivalence between the sum of symmetrized KL information between the output hypothesis and individual samples and that between the hypothesis and the entire dataset.
We prove this by providing explicit expressions for the gap between these measures in the non-asymptotic regime. 
Additionally, we characterize the asymptotic behavior of various information measures in the context of the Gibbs algorithm, leading to tighter generalization error bounds. 
An illustrative example is provided to verify our theoretical results, demonstrating our analysis holds in broader settings.

\end{abstract}


\section{Introduction}\label{:Sec1}



One of the most important research topics in statistical learning theory is to capture the generalization behavior of the learning algorithms to avoid overfitting. Recently, \cite{russo2019much,xu2017information} proposed an information-theoretic approach to bound generalization error, where a learning algorithm is modeled as a randomized channel that takes the training dataset as input and outputs the learned hypothesis.
In this setting, different information measures can be used to derive various non-trivial generalization error bounds, which capture all components in supervised learning, including the data-generating distribution, hypothesis class, and the learning algorithm itself. In comparison, traditional approaches such as VC-dimension \cite{788640}, algorithmic
stability \cite{bousquet2002stability}, algorithmic robustness \cite{xu2012robustness}, and PAC-Bayesian bounds \cite{mcallester2003pac} cannot exploit all the aspects that affect the generalization performance.

After the seminal work~\cite{xu2017information}, several  approaches~\cite{asadi2018chaining,hellstrom2020generalization,hafez2020conditioning,haghifam2020sharpened,steinke2020reasoning,haghifam2021towards,harutyunyan2021information,aminian2022tighter,zhou2023stochastic} have been proposed to refine information-theoretic generalization error bounds. Among them, a significant advancement is presented in~\cite{bu2020tightening}, where the individual sample mutual information bound is introduced. By focusing on information measures involving individual samples, this bound is not only tighter but also simplifies the empirical estimation process, for instance, by using neural estimators like MINE \cite{pmlr-v80-belghazi18a}. In contrast, the mutual information-based bound in \cite{xu2017information} depends on the mutual information between the hypothesis and the entire training dataset, making it nearly impossible to estimate with a large sample size $n$.


This paper explores a similar individual sample approach in the context of a specific learning algorithm, 
the Gibbs algorithm (formally defined in~\eqref{eq:Gibbs}). Such an algorithm can be interpreted as a randomized variant of the standard empirical risk minimization algorithm with mutual information as regularization, and it has other important connections to  SGLD~\cite{raginsky2017non} and PAC-Bayesian bound \cite{thiemann2017strongly}. More importantly, it has been shown in~\cite{aminian2021exact,aminian2023information,chen2024gibbs} that the generalization error of the Gibbs algorithm can be characterized exactly using the symmetrized KL information between the hypothesis and the entire dataset. Just like the mutual information-based bound, this exact characterization also suffers from the same drawback in practical evaluation due to its high dimensionality. 



To address such an issue, in this paper, we study the individual sample information measures within the Gibbs algorithm and their counterparts involving the entire dataset. Our main contribution is an equivalency between the sum of symmetrized KL information between the output hypothesis and individual samples and that between the hypothesis and the entire dataset in the asymptotic regime $n\to \infty$. We also present other interesting properties of information measures in the individual sample context for both non-asymptotic and asymptotic regimes. In particular:
\begin{itemize}
    \item In Section~\ref{:Sec3}, we provide an explicit expression of the gap between the sum of symmetrized KL information w.r.t individual samples and that w.r.t the entire dataset in the non-asymptotic regime.
    \item In Section~\ref{:Sec4}, we precisely characterize the asymptotic behavior of different information measures for the Gibbs algorithm in terms of both convergence rate and constant factor. We then present our main theorem and additional results derived using similar techniques, which leads to a tighter bound on the generalization error.
    \item In Section~\ref{:Sec5}, an illustrative mean estimation example is provided to verify all theoretical results and demonstrate that our findings can hold in a more general setting.
\end{itemize}

\section{Preliminaries}\label{:Sec2}
In this section, we first introduce some background about information measures and the Gibbs algorithm. 

\subsection{Relevant Information Measures}
Given two probability measures $P$ and $Q$ defined on the same probability space $(\Omega,\mathcal{F})$, 
the symmetrized Kullback-Leibler (KL) divergence is defined as
\begin{equation}
    D_{\mathrm{SKL}}(P\Vert Q) \triangleq D(P\Vert Q)+D(Q\Vert P),
\end{equation}
which symmetrizes the standard KL divergence $D(P\Vert Q)$. When $P\ll Q\ll P$ where $\ll$ denotes absolute continuity between measures, symmetrized KL divergence can be written as
\begin{equation}
    D_{\mathrm{SKL}}(P\Vert Q)=\mathbb{E}_{Q}\bigg[\frac{dP}{dQ}\log\frac{dP}{dQ}-\log\frac{dP}{dQ}\bigg].
\end{equation}
It is natural to see that symmetrized KL divergence also belongs to the $f$-divergence family~\cite{sason2016f}.


For two random variables $X$ and $Y$, their mutual information is the KL divergence between their joint distribution and the product of the marginal distributions, i.e., $I(X;Y) \triangleq D(P_{X,Y}\Vert P_X\otimes P_Y)$.
Similarly, we can define the symmetrized KL information as
\begin{equation}\label{equ: Iskl=I+L}
\begin{aligned}
    I_{\mathrm{SKL}}(X;Y) &\triangleq D_{\mathrm{SKL}}(P_{X,Y}\Vert P_X\otimes P_Y)\\
    &=I(X;Y)+L(X;Y),
\end{aligned}
\end{equation}
where $L(X;Y)\triangleq D( P_X\otimes P_Y \Vert P_{X,Y})$ represents lautum information~\cite{palomar2008lautum}.

\subsection{Generalization Error in Supervised Learning}
We denote ${\mathcal{W}}$ as the hypothesis class and ${\mathcal Z}$ as the instance space. A training dataset $S=\{Z_i\}_{i=1}^n \in \mathcal{S}$ with $Z_i\in { \mathcal Z}$ consists $n$ samples drawn i.i.d from the data-generating distribution $\mu$. A loss function $\ell:{ \mathcal W}\times{ \mathcal Z}\to \mathbb{R}_0^+$ is used to measure the performance of a hypothesis on a sample $Z$. Therefore, we define the empirical and population risks associated with a given hypothesis $w$ by
\begin{align}
    L_e(w,s) &\triangleq  \frac{1}{n}\sum_{i=1}^n\ell(w,z_i), \\
    L_\mu(w)& \triangleq \mathbb{E}_{Z\sim \mu}[\ell(w,Z)],
\end{align}
respectively. In statistical learning, a learning algorithm can be modeled as a randomized mapping from the training set $S$ onto a hypothesis $W\in\mathcal{W}$ according to the conditional distribution $P_{W|S}$. We define the expected generalization error quantifying the degree of over-fitting as
\begin{align}
        {\rm{gen}}(P_{W|S},P_S) &\triangleq \mathbb{E}_{P_{W,S}}[L_\mu(W)-L_e(W,S)]\\
        &=\mathbb{E}_{P_{W}\otimes P_{S}}[L_e(W,S)]-\mathbb{E}_{P_{W,S}}[L_e(W,S)], \nn
\end{align}
where the joint distribution $P_{W,S} =  P_{W|S}\otimes P_S = P_{W|S}\otimes \mu^n$.

Following the framework proposed in~\cite{aminian2023information}, we focus on a specific learning algorithm $P_{W|S}$, i.e., the Gibbs algorithm (or Gibbs posterior~\cite{catoni2007pac}), which is defined as
\begin{equation}\label{eq:Gibbs}
    P^{[n]}_{W|S}(w|s) \triangleq \frac{\pi(w)e^{-\gamma L_e(w,s)}}{V_{L_e}(s,\gamma)}.
\end{equation}
Here, $\gamma$ is the inverse temperature, $\pi(w)$ is an arbitrarily chosen prior distribution over ${ \mathcal W}$, and
$V_{L_e}(s,\gamma) \triangleq \int_{\mathcal{W}}\pi(w)e^{-\gamma L_e(w,s)}dw$ is the partition function that normalizes the distribution.

As shown in~\cite{aminian2021exact,aminian2023information}, an important property of the Gibbs algorithm is that its generalization error can be exactly characterized using the symmetrized KL information:
\begin{equation}
    {\rm{gen}}(P_{W|S},P_S) = I_{\mathrm{SKL}}(W;S)/\gamma.
\end{equation}

\subsection{Other Notations}
We will adopt the following notations to express the asymptotic scaling of quantities with $n$: $f(n) = O(g(n))$ represents that there exists a constant $c$ s.t. $|f (n)| \le c g(n)$; $f(n)=\Theta(g(n))$ when there exist two constants $c_1>0$, $c_2>0$ s.t. $c_1g(n)\le f(n)\le c_2g(n)$; $f(n)=o(g(n))$ when $\lim_{n\to\infty}(f(n)/g(n))=0$; and $f(n)\sim g(n)$ when $\lim_{n\to\infty}(f(n)/g(n))=1$.

To simplify notation, we denote a probability measure or its corresponding probability density function by $P_{W}$ when there is no ambiguity.
We use $P_{W|Z^n}$ to represent the conditional probability density function, with the capital $W,Z$ representing that it is also a random variable. 



Throughout the paper, we will consider the Gibbs algorithm with a fixed inverse temperature $\gamma$ and study its asymptotic behavior as the number of training samples $n \to \infty$. 
It is convenient for us to define the Gibbs algorithm using the population risk, i.e.,
\begin{equation}
    P_{W}^{\infty}(w) \triangleq \frac{\pi(w)e^{-\gamma L_\mu(w)}}{\int_{\mathcal{W}}\pi(w)e^{-\gamma L_\mu(w)}dw},
\end{equation}
and the expectation of any measurable function $f(\cdot)$ under $P_{W}^{\infty}$ is denoted as
\begin{equation}
        \mathbb{E}_{W}^{\infty}[f(W)] \triangleq \int_{\mathcal{W}} P_{W}^{\infty}(w) f(w)dw.
\end{equation}




\section{Non-asymptotic Results}\label{:Sec3}
Motivated by the idea of using individual sample information measures proposed in~\cite{bu2020tightening}, we first investigate the connection between the joint symmetrized KL information and its individual sample counterpart for the Gibbs algorithm. 

The following theorem states that the difference between these two information measures can be characterized using the Jensen gap. 
\begin{theorem}\label{:thm1}
    For joint distribution $P_{W,S}$ induced by the Gibbs algorithm, we have
    \begin{align}\label{equ:Th1 gap}
            &\sum_{i=1}^nI_{\mathrm{SKL}}(W;Z_i)-I_{\mathrm{SKL}}(W;S)\nn\\            =&\sum_{i=1}^n\bigg(\!\mathbb{E}_{P_{W,Z_i}}[J_{i}^{[n]}(W,Z_i)]\!-\!\mathbb{E}_{P_{W}\otimes P_{Z_i}}[J_{i}^{[n]}(W,Z_i)]\!\bigg), 
    \end{align}
where the Jensen gap $J_i^{[n]}(w,z_i)$ is defined as
    \begin{align}\label{equ:Th1_def1}
        J_i^{[n]}(w,z_i)&\triangleq \log\int_{\mathcal{Z}^{n-1}}P^{[n]}_{W|S}(w|z_i,z^{-i}) d\mu^{n-1}(z^{-i})
        \\&\quad -\int_{\mathcal{Z}^{n-1}}\log \Big(P^{[n]}_{W|S}(w|z_i,z^{-i})\Big) d\mu^{n-1}(z^{-i}), \nn
    \end{align}
  with $z^{-i} \triangleq \{z_1,\cdots,z_{i-1},z_{i+1},\cdots,z_n \}$. 
\end{theorem}
See Appendix~\ref{appendixA} for the proof.
We note that this theorem holds whenever the samples $S$ are drawn independently but not necessarily identically generated from the distribution $\mu$. 

\begin{remark}
As the $\log$ function is concave, the Jensen gap $J_i^{[n]}(w,z_i)$ is always non-negative. However, the RHS of \eqref{equ:Th1 gap} can be either negative or positive. An example showing that $I_{\mathrm{SKL}}(W;S)$ can be either larger or smaller than $\sum_{i=1}^nI_{\mathrm{SKL}}(W;Z_i)$
can be found in~\cite[Example 1]{aminian2021exact}.
\end{remark}

It is worth mentioning that the Jensen gap $J_i^{[n]}(w,z_i)$ in Theorem~\ref{:thm1} has its own operational meaning by making the connection to the worst-case data-generating distribution introduced in~\cite{zou2024worst}. A detailed discussion can be found in Appendix~\ref{appendixA}.
Other than this, interpreting this Jensen gap directly through finite sample analysis is challenging, prompting us to delve into the asymptotic regime in the next section.

\section{Asymptotic Results}\label{:Sec4}
In this section, we provide an asymptotic analysis of different information measures with i.i.d samples, e.g., the joint symmetrized KL information
and its individual sample counterpart for the Gibbs algorithm.

\subsection{Asymptotics of Individual Sample Information Measures}
We start by rigorously defining the limiting probability space $({\mathcal W}\times {\mathcal Z}^\infty,\mathcal{F}^\infty,P_{W}^{\infty}\otimes P_{Z^\infty})$ in the following definition. 
\begin{definition}
    As the training data were i.i.d sampled from data distribution $\mu$, 
    there exists a filtered probability space $({\mathcal Z}^\infty,\mathcal{F}_{Z^\infty},\{\mathcal{F}_{Z^n}^{[n]}\},P_{Z^\infty})$ where
    \begin{equation}    \mathcal{F}_{Z^n}^{[n]}=\sigma(Z_1,Z_2,\dots,Z_n),\     \mathcal{F}_{Z^\infty}=\sigma\bigg(\bigcup_n\mathcal{F}_{Z^n}^{[n]}\bigg).
    \end{equation}
    We define a probability space $(\mathcal{W},\mathcal{B},P_{W}^{\infty})$
    and the following product probability space 
    \begin{align}
    &({\mathcal W}\times {\mathcal Z}^\infty,\mathcal{F}^\infty,\{\mathcal{F}_{W,Z^n}^{[n]}\},P_{W}^{\infty}\otimes P_{Z^\infty})\nn\\
    \triangleq &({\mathcal W},\mathcal{B},P_{W}^{\infty})\times({\mathcal Z}^\infty,\mathcal{F}_{Z^\infty},\{\mathcal{F}_{Z^n}^{[n]}\},P_{Z^\infty}).
    \end{align}
  For every sub-$\sigma$-algebra $\mathcal{F}_{Z^n}^{[n]}$, $P_{W,Z^n}^{[n]}$ is the probability measure induced by the Gibbs algorithm and the distribution of the dataset with size $n$, and $P_{W,Z_i}^{[n]}$ is the marginalization of $P_{W,Z^n}^{[n]}$.
\end{definition}

Now, we are ready to study the asymptotic behavior of the Gibbs algorithm. We start by presenting the following two lemmas that capture the limit of the joint distribution $P_{W,Z^n}^{[n]}$.
\begin{lemma}
For non-negative loss $\ell(w,z)\ge 0$, we have
\begin{equation}
    \lim_{n\to\infty}\left(\frac{dP_{W,Z^n}^{[n]}}{dP_{W}^{\infty}\otimes P_{Z^\infty}}\right)=1\ \ a.s.
\end{equation}
\end{lemma}
\begin{lemma} 
For non-negative loss $\ell(w,z)\ge 0$ and any individual sample $Z_i$, we have
\begin{equation}        \varliminf_{n\to\infty}\left(\frac{dP_{W,Z_i}^{[n]}}{dP_{W}^{\infty}\otimes P_{Z^\infty}}\right)=1\ \ a.s.
\end{equation}
\end{lemma}

These two lemmas rigorously confirm the intuition that as $n\to\infty$, the asymptotic joint distribution $P_{W,Z^n}^{[n]}$ will converge to a product measure $P_{W}^{\infty}\otimes P_{Z^\infty}$, i.e., the learned hypothesis $W$ depends solely on the data distribution $\mu$ and is independent of the dataset $S$. It is worth mentioning that this result is widely applicable, as it only requires the loss function to be non-negative or lower-bounded.


\begin{corollary}\label{:Cor1}
    If we further assume that the loss function is bounded, i.e., $\ell(w,z)\in [0, C]$, we have that $\big(\frac{dP_{W,Z^n}^{[n]}}{dP_{W}^{\infty}\otimes P_{Z^\infty}}\big)$ and $\big(\frac{dP_{W,Z_i}^{[n]}}{dP_{W}^{\infty}\otimes P_{Z^\infty}}\big)$ are both uniformly bounded. Furthermore, $\lim_{n\to\infty}\big(\frac{dP_{W,Z_i}^{[n]}}{dP_{W}^{\infty}\otimes P_{Z^\infty}}\big)=1$ almost surely.
\end{corollary}


In the following, we will focus on the bounded loss function case. We already know that $dP_W^{[n]}\otimes P_{Z_i}/dP_{W,Z_i}^{[n]}$ converges to 1 as $n\to\infty$, and the following lemma characterizes the exact rate of such convergence. 


\begin{lemma}\label{:Lem3}
If the loss function $\ell(w,z)$ is bounded, we have 
    \begin{align}\label{equ:lem3_aim}
        &\lim_{n\to\infty}n\cdot\bigg(1-\frac{dP_W^{[n]}\otimes P_{Z_i}}{dP_{W,Z_i}^{[n]}}\bigg)\\=&-\gamma[\ell(W,Z_i)-L_\mu(W)]+\mathbb{E}_{W}^{\infty}[\gamma(\ell(W,Z_i)-{L_\mu(W)})]. \nn
    \end{align}
    Additionally, $\displaystyle n\cdot\bigg(1-\frac{dP_W^{[n]}\otimes P_{Z_i}}{dP_{W,Z_i}^{[n]}}\bigg)$ is uniformly bounded.
\end{lemma}

Built upon this Lemma, we provide the following theorem that characterizes the convergence rate of $I_{\mathrm{SKL}}(W;Z_i)$ with a tight constant factor as $n\to \infty$.

\begin{theorem}\label{:Thm2}
If the loss function $\ell(w,z)$ is bounded, we have
\begin{align}\label{sim:Iskl}
    I_{\mathrm{SKL}}(W;Z_i) & \sim\frac{\gamma^2}{n^2}\mathbb{E}_{\mu}\bigg[\mathbb{E}_{W}^{\infty}\big[(\ell(W,Z)-L_\mu(W))^2\big]\nn\\&\qquad \ \ \quad-\mathbb{E}_{W}^{\infty}\big[(\ell(W,Z)-L_\mu(W))\big]^2\bigg]. 
\end{align}
\end{theorem}

The constant on the right-hand side of \eqref{sim:Iskl} can be interpreted as the variance of $\ell(W,Z)-L_\mu(W)$, which is always non-negative. It is also strictly positive unless $\ell(w,z)$ is a constant for every fixed $w$. 




The proof of Lemma~\ref{:Lem3} and Theorem~\ref{:Thm2} mainly use the strong law of large numbers and the dominated convergence theorem, and more details can be found in Appendix~\ref{appendixB}.
As shown in the following corollaries, the same technique can also be applied to other information measures, specifically, the $\chi^2$ information.
\begin{corollary}
The $\chi^2$ information has the similar rate if the loss function $\ell(w,z)$ is bounded, i.e.,
    \begin{equation}
        I_{\chi^2}(W;Z_i)= \Theta\bigg(\frac{1}{n^2}\bigg),
    \end{equation}
    furthermore,
    \begin{equation}
        I_{\chi^2}(W;Z_i)\sim I_{\mathrm{SKL}}(W;Z_i).
    \end{equation}
\end{corollary}
\begin{corollary}\label{:Cor2.5}
If the loss function $\ell(w,z)$ is bounded, we have
    \begin{equation}
        I(W;Z_i)= O\bigg(\frac{1}{n^2}\bigg).
    \end{equation}
\end{corollary}
\begin{remark}\label{:Rmk1.5}
Corollary~\ref{:Cor2.5} is directly obtained using Theorem~\ref{:Thm2} and the fact that $I(W;Z_i)\leq I_{\mathrm{SKL}}(W;Z_i)$. 
However, if we directly use the bounding technique used in Theorem~\ref{:Thm2} for mutual information, it yields a weaker conclusion $I(W;Z_i)=o\big(\frac{1}{n}\big)$. One possible reason is that mutual information corresponds to f-divergence with $f(x) = x\log x$, which is not consistently positive for all $x >0$. On the other hand, $f(x) = x\log x-\log x$ for symmetrized KL information, which is always non-negative. Therefore, swapping the expectation and the limit in the proof of Theorem~\ref{:Thm2} will have a smaller impact on the analysis, leading to a more refined characterization.
The same argument applies to $\chi^2$ divergence as well.
\end{remark}


\subsection{Asymptotics of the Gap}
In this subsection, we focus on the gap between the sum of $I_{\mathrm{SKL}}(W;Z_i)$ and $I_{\mathrm{SKL}}(W;S)$ in the asymptotic regime. Our goal is to prove that this gap converges to zero faster than $I_{\mathrm{SKL}}(W;S)$, i.e., the generalization error itself. We begin by presenting two lemmas that capture the asymptotic behaviors of the Jensen gap $J_i^{[n]}(w,z_i)$ defined in Thorem~\ref{:thm1}. 
\begin{lemma}\label{:Lem4}
 If the loss function $\ell(w,z)$ is bounded, there exists a sequence of functions $\hat{J}^{[n]}(w)$ independent of $z_i$ such that
    \begin{equation}        \lim_{n\to\infty}n\cdot(\hat{J}^{[n]}(w)-J_i^{[n]}(w,z_i))=0.
    \end{equation}
    Furthermore, $n\cdot(\hat{J}^{[n]}(w)-J_i^{[n]}(w,z_i))$ is uniformly bounded.
\end{lemma}

This result shows that despite $J_i^{[n]}(w,z_i)$ is a function of both $w$ and $z_i$, the influence of $z_i$ is relatively negligible when $n$ goes to infinity.

\begin{lemma}\label{:Lem5}
If the loss function $\ell(w,z)$ is bounded, the $\hat{J}^{[n]}(w)$ introduced in Lemma~\ref{:Lem4} satisfying $n\cdot\hat{J}^{[n]}(w)$ is uniformly bounded. Furthermore, $\lim_{n\to\infty}n\cdot\hat{J}^{[n]}(w)$ exists.
\end{lemma}

Equipped with these technical lemmas, we present the main theorem of the paper, which shows that the gap between $I_{\mathrm{SKL}}(W;S)$ and the sum of $I_{\mathrm{SKL}}(W;Z_i)$  converges to zero faster than $\frac{1}{n}$. 

\begin{theorem}\label{:Thm4}
If the loss function $\ell(w,z)$ is bounded, we have
\begin{equation}\label{equ:Th4 IsklWZ-WS}
    \sum_{i=1}^nI_{\mathrm{SKL}}(W;Z_i)-I_{\mathrm{SKL}}(W;S)= o\Big(\frac{1}{n}\Big).
\end{equation}
\end{theorem}

It is worth pointing out that the key for proving Theorem~\ref{:Thm4} is Lemma~\ref{:Lem4}, which indicates that the effect of terms involving variable $z_i$ is order-wise small compared to the remaining terms. Utilizing this result, we apply the dominated convergence theorem on  $n \cdot \big(\sum_{i=1}^nI_{\mathrm{SKL}}(W;Z_i)-I_{\mathrm{SKL}}(W;S)\big)$. More proof details can be found in Appendix~\ref{appendixB}.


From Theorem~\ref{:Thm4}, we can immediately obtain the following corollary. 
\begin{corollary}\label{:Cor3}
If the loss function $\ell(w,z)$ is bounded, we have
    \begin{equation}
        I_{\mathrm{SKL}}(W;S)= \Theta\bigg(\frac{1}{n}\bigg).
    \end{equation}
    More specifically,
    \begin{equation}
        I_{\mathrm{SKL}}(W;S)\sim \sum_{i=1}^n  I_{\mathrm{SKL}}(W;Z_i).
    \end{equation}
\end{corollary}


The result in Corollary~\ref{:Cor3} aligns with the conclusion in~\cite{raginsky2016information}, which states that for a Gibbs algorithm, when the loss function $\ell\in[0,1]$,
\begin{equation}\label{equ: example bound}
    |{\rm{gen}}(P_{W|S}^{[n]},P_S)|\leq \frac{\gamma}{2n}.
\end{equation}

\begin{remark}
As a sanity check, we look at a simple coin-tossing example, where $w\in\{0,1\}$ and $z\in\{0,1\}$,  $\ell(w,z) = \mathds{1}_{w=z}$ is a zero-one loss, and $\pi(w)$ is uniform over $\{0,1\}$. 
From Corollary~\ref{:Cor3}, the convergence behavior of $I_{\mathrm{SKL}}(W;S)$ and thus ${\rm{gen}}(P_{W|S}^{[n]},P_S)$ can be calculated as
\begin{equation}
    \lim_{n\to\infty}n\cdot {\rm{gen}}(P_{W|S}^{[n]},P_S)=\frac{\gamma}{4},
\end{equation}
which indicates that the $\gamma/2n$ bound  in \eqref{equ: example bound} is not tight.
\end{remark}

From Corollary~\ref{:Cor3}, it is easy to see that $I(W;S)\leq I_{\mathrm{SKL}}(W;S)=\Theta\big(\frac{1}{n}\big)$, indicating $I(W;S)=O\big(\frac{1}{n}\big)$. With the same argument, the lautum information also satisfies that $L(W;S)=O\big(\frac{1}{n}\big)$. However, it is not clear which quantity contributes more to the generalization error of the Gibbs algorithm. The following theorem answers the question by proving that the two information measures equal each other asymptotically.
\begin{theorem}\label{:Thm5}
If the loss function $\ell(w,z)$ is bounded, we have
    \begin{equation*}
        \lim_{n\to\infty}n\cdot I(W;S)=\lim_{n\to\infty}n\cdot L(W;S)=\frac{1}{2}\lim_{n\to\infty}n\cdot I_{\mathrm{SKL}}(W;S).
    \end{equation*}
     In other words, $I(W;S)\sim L(W;S)=\Theta\big(\frac{1}{n}\big)$.
\end{theorem}
The proof technique of Theorem~\ref{:Thm5} differs from those used in Lemma~\ref{:Lem3} and Theorem~\ref{:Thm2}. Here, our idea is to sandwich the target quantity between an upper bound and a lower bound that differ only in the third or higher-order terms.
We then prove that the two bounds converge to the same value, characterized by the second-order term. Instead of using the dominated convergence theorem as in Lemma~\ref{:Lem3} and Theorem~\ref{:Thm2}, we directly analyze the integral of the second-order terms for any $n$ before taking the limit. In this process, the independence of the samples plays a crucial role, ensuring that all interaction terms are zero.



Using Theorem~\ref{:Thm5}, we can provide an alternative proof for Corollary~\ref{:Cor2.5} by applying the Proposition 2 of~\cite{bu2020tightening}, i.e.,
\begin{equation}
    I(W;S)\geq \sum_{i=1}^nI(W;Z_i).
\end{equation}


We provide the following result to showcase Theorem~\ref{:Thm5},  which tightens the existing generalization error bound for the Gibbs algorithm.
\begin{theorem}
    For Gibbs algorithm with bounded loss function $\ell(w,z)\in [a,b]$, $\forall \delta>0$, there exist an $N\in \mathbb{N}^+$ such that $\forall n>N$,
    \begin{equation}
        0\leq {\rm gen}(P_{W|S}^{[n]},P_S)\leq \frac{(b-a)^2\gamma}{(4-\delta)n}.
    \end{equation}
\end{theorem}
This theorem provides a tighter bound compared with \eqref{equ: example bound}. Revisiting the coin-tossing example, it is interesting to see that this bound is asymptotically tight in this circumstance.

\subsection{Comparison with Asymptotics
of Model Capacity}

We would like to compare our result with the asymptotic model capacity studied in~\cite{rissanen1984universal,clarke1990information}. Different from our setting, they considered data $Y^n\in{\mathcal Y}^n$ that are drawn i.i.d from $P_{Y|X}(y|x)$, where $X\in{\mathcal X} \subset \mathbb{R}$ denotes the model parameter from certain model family. For such a Bayesian setting, the model parameter $X$ is modeled using a prior distribution $P_X(x)$. If the model $P_{Y|X}$ is sufficiently smooth in $X$, then
\begin{align}
    I(X;Y^n) &= \frac{1}{2}\log\frac{n}{2\pi e}-D(P_X\Vert P_X^*)-\log\int_{\mathcal X}\sqrt{J(x')}dx' \nn\\
    &\qquad + o(1),
\end{align}
where
\begin{equation}   J(x)=\mathbb{E}_{P_{Y|X=x}}\bigg[\bigg(\frac{\partial}{\partial x}\log P_{Y|X}(Y|x)\bigg)^2\bigg],
\end{equation}
and $P_X^*$ denotes the least informative prior, i.e., Jeffery's prior~\cite{clarke1994jeffreys}. The mutual information is maximized when $P_X = P_X^*$, i.e., $P_X^*$ is the capacity achieving distribution. 
We can see the growing rate of mutual information is $I(X;Y^n) = O(\log n)$, which is different from the asymptotic result $I(W;S) = \Theta(\frac{1}{n})$ in Theorem~\ref{:Thm5}.

The difference between our setting and the model capacity setting is two-fold: 1) we considered i.i.d samples from data distribution $\mu$, while the model capacity setting considers samples conditionally independent generated from $P_{Y|X}$; 2) In our setting, the channel $P_{W|Z^n}$ is the  Gibbs algorithm defined using a bounded loss function $\ell(w,z)$, so that conditional on $W$, the samples are not independent to each other anymore. However, the posterior $P_{X|Y^n}$ in the model capacity setting is induced by the prior $P_{X}$ and the likelihood $P_{Y|X}$ given the conditional independent structure among the samples. 


Note that the assumption $\ell(W,Z_i)$ being bounded is a sufficient condition for our previous results and is adopted to circumvent the technical difficulty of exchanging the order of integration and Taylor expansion as well as the order of integration and limits. 

\section{Example}\label{:Sec5}
In this section, we will consider an example beyond the bounded loss function. It can be shown that most of our conclusions still hold under this setting. A more detailed elaboration of the example can be found in Appendix~\ref{appendixC}.

\textbf{Estimating the mean.} 
Let $S=\{Z_i\}_{i=1}^n$ be the training set, where $Z_i$ is a $d$ dimensional vector sampled i.i.d. from $\mu = \mathcal{N}(0_d,(\frac{1}{\sqrt{2\beta}})^2 \mI_d)$. We consider the problem of learning the means of the distribution $\mu$. For simplicity, we consider $d=1$. We adopt square error $\ell(w,Z) \triangleq \Vert w-Z\Vert_2$ as the loss function and choose our prior distribution to be $\pi(w)=\frac{1}{\sqrt{\pi}}\exp(-w^2)$.

In this simple example, we can calculate the joint symmetrized KL information between $S$ and $W$ as
\begin{equation}\label{equ:res_IsklWS}
    \begin{aligned}
    \gamma{\rm{gen}}(P_{W|S},\mu) 
    =I_{\mathrm{SKL}}(W;S)=
    \frac{\gamma^2}{n\beta(1+\gamma)},
    \end{aligned}
\end{equation}
and the individual sample symmetrized KL information between $W$ and $Z_i$
\begin{equation}\label{equ:res_IsklWZ}
    \begin{aligned}
    I_{\mathrm{SKL}}(W;Z_i)
    &=\frac{\gamma^2}{n^2\beta(1+\gamma)+\gamma^2(n-1)}.
    \end{aligned}
\end{equation}
From \eqref{equ:res_IsklWS} and \eqref{equ:res_IsklWZ}, we get
\begin{align}\label{equ:gap_example}
\sum_{i=1}^nI_{\mathrm{SKL}}(W;Z_i)-I_{\mathrm{SKL}}(W;S) &= \Theta\bigg(\frac{1}{n^2}\bigg)  \\
&=o(I_{\mathrm{SKL}}(W;S)), \nn
\end{align}
which shows that Theorem~\ref{:Thm4} still holds, despite the fact we are not considering a bounded loss function.


We further investigate the Jensen gap term $J_i(w,z_i)$, since as stated previously, the key to prove Theorem~\ref{:Thm4} is that the effect of variable $z_i$ is order-wise smaller than that of $w$. In estimating the mean problem, we can calculate that
\begin{equation}
    \begin{aligned}
    J_i^{[n]}(w,z_i)
    =& w^2\Theta\bigg(\frac{1}{n}\bigg)+wz_1\Theta\bigg(\frac{1}{n^2}\bigg)\\&+z_1^2\Theta\bigg(\frac{1}{n^3}\bigg)+\Theta\bigg(\frac{1}{n^2}\bigg)
    \end{aligned}
\end{equation}
where all terms represented by big O notation were uniformly small. We can see that the contribution of $z_i$ term is $\Theta\big(\frac{1}{n^2}\big)$, which is indeed order-wise smaller than those terms not containing $z_i$.

Finally, we provide a similar analysis of mutual information.
\begin{equation}
    \begin{aligned}
        I(W;Z_i)&=\frac{1}{2}\log\bigg(1+\frac{\gamma^2}{n^2(1+\gamma)\beta+(n-1)\gamma^2}\bigg)\\
        &=\Theta\bigg(\frac{1}{n^2}\bigg),
    \end{aligned}
\end{equation}
and
\begin{equation}
    \begin{aligned}
        I(W;S)&=\frac{1}{2}\log\bigg(1+\frac{\gamma^2}{n\beta(1+\gamma)}\bigg)\\
        &\sim \frac{1}{2}\cdot\frac{\gamma^2}{n\beta(1+\gamma)}\\
        &=\frac{1}{2}I_{\mathrm{SKL}}(W;S),
    \end{aligned}
\end{equation}
which corresponds to our result in Theorem~\ref{:Thm5}.

\newpage



\bibliography{ref}
\bibliographystyle{ieeetr}

\onecolumn
\appendices
\section{Non-asymptotic Results}\label{appendixA}
\begin{notation} For simplicity, we use
    \begin{equation}
        \mathbb{E}_{e}^{[n]}[f(W)]\triangleq\frac{\int_{\mathcal{W}}\pi(w)e^{-\gamma \frac{1}{n}\sum_{i=1}^n\ell(w,z_i)}(f(w))dw}{\int_{\mathcal{W}}\pi(w)e^{-\gamma \frac{1}{n}\sum_{i=1}^n\ell(w,z_i)}dw}.
    \end{equation}
\end{notation}
\subsection{Proof of Theorem 1}
\begin{proof}
    First from the definition, we have
    \begin{equation}\label{equ:Th1_def2}
    \begin{aligned}
    I_{\mathrm{SKL} }(W;Z_i)&=\mathbb{E}_{P_{W,Z_i}}[\log(P_{W|Z_i})]-\mathbb{E}_{P_{W}\otimes P_{Z_i}}[\log(P_{W|Z_i})].\\
    \end{aligned}
    \end{equation}
    Using $z^{-i}$ to represent $\{z_j\}_{j\ne i}$, the inner term $\log(P_{W|Z_i})$ can be represented by:
    \begin{equation}\label{equ:Th1_eq1}
    \begin{aligned}
    \log(P_{W|Z_i})&=\log\left(\int_{\mathcal{Z}^{n-1}}P_{W|Z_i,z^{-i}}d\mu^{n-1}(z^{-i}) \right).\\
    \end{aligned}
    \end{equation}
    Notice that this holds true only when $Z_i$s are sampled independently.

    Our aim is to prove
    \begin{equation}
        \sum_{i=1}^nI_{\mathrm{SKL} }(W;Z_i)=\sum_{i=1}^n(\mathbb{E}_{P_{W,Z_i}}[J_i^{[n]}(W,Z_i)]-\mathbb{E}_{P_{W}\otimes P_{Z_i}}[J_i^{[n]}(W,Z_i)])+I_{\mathrm{SKL} }(W;S).
    \end{equation}
    Then from \eqref{equ:Th1_def1} \eqref{equ:Th1_def2} and \eqref{equ:Th1_eq1}, the equation to prove can be rewritten as
    \begin{align}\label{eq:Th1_aim}
        I_{\mathrm{SKL} }(W;S) =&\  \sum_{i=0}^n\bigg[ \mathbb{E}_{P_{W,Z_i}}\left[\int_{\mathcal{Z}^{n-1}}\log(P_{W|Z_i,z^{-i}})d\mu^{n-1}(z^{-i})\right]-\nonumber\\&\ \mathbb{E}_{P_{W}\otimes P_{Z_i}}\left[\int_{\mathcal{Z}^{n-1}}\log(P_{W|Z_i,z^{-i}})d\mu^{n-1}(z^{-i})\right] \bigg].
    \end{align}

    Using the Gibbs posterior, we have
    \begin{equation}
        P_{W|Z_i,z^{-i}}=\frac{\pi(W)e^{-\gamma L_e(W,Z_i,z^{-i})}}{V_{L_e}(Z_i,z^{-i},\gamma)}.
    \end{equation}
    With this, we can now calculate the right hand side of equation \eqref{eq:Th1_aim}
    \begin{align}\label{equ:gap=gen}
        &\ \sum_{i=0}^n\left[ \mathbb{E}_{P_{W,Z_i}}\left[\int_{\mathcal{Z}^{n-1}}\log(P_{W|Z_i,z^{-i}})d\mu^{n-1}(z^{-i})\right]-\mathbb{E}_{P_{W}\otimes P_{Z_i}}\left[\int_{\mathcal{Z}^{n-1}}\log(P_{W|Z_i,z^{-i}})d\mu^{n-1}(z^{-i})\right] \right]\nonumber\\
        =&\ \sum_{i=0}^n\left\{ \left[\int_{\mathcal{Z}^{n-1}}\mathbb{E}_{P_{W,Z_i}}[\log(P_{W|Z_i,z^{-i}})]d\mu^{n-1}(z^{-i})\right]-\left[\int_{\mathcal{Z}^{n-1}}\mathbb{E}_{P_{W}\otimes P_{Z_i}}[\log(P_{W|Z_i,z^{-i}})]d\mu^{n-1}(z^{-i})\right] \right\}\nonumber\\
        =&\ \sum_{i=0}^n \int_{\mathcal{Z}^{n-1}}\left[\mathbb{E}_{P_{W,Z_i}}[\log(P_{W|Z_i,z^{-i}})]-\mathbb{E}_{P_{W}\otimes P_{Z_i}}[\log(P_{W|Z_i,z^{-i}})]\right]d\mu^{n-1}(z^{-i})\nonumber\\
        =&\ \sum_{i=0}^n \int_{\mathcal{Z}^{n-1}}\gamma\left[\mathbb{E}_{P_{W}\otimes P_{Z_i}}[ L_e(W,Z_i,z^{-i})]-\mathbb{E}_{P_{W,Z_i}}[L_e(W,Z_i,z^{-i})]\right]d\mu^{n-1}(z^{-i})\nonumber\\
        =&\ \sum_{i=0}^n \int_{\mathcal{Z}^{n-1}}\frac{\gamma}{n}\left[\mathbb{E}_{P_{W}\otimes P_{Z_i}}[ \ell(W,Z_i)]-\mathbb{E}_{P_{W,Z_i}}[\ell(W,Z_i)]\right]d\mu^{n-1}(z^{-i})\nonumber\\
        =&\ \sum_{i=0}^n \frac{\gamma}{n}\left[\mathbb{E}_{P_{W}\otimes P_{Z_i}}[ \ell(W,Z_i)]-\mathbb{E}_{P_{W,Z_i}}[\ell(W,Z_i)]\right]\nonumber\\
        =&\ \sum_{i=0}^n \frac{\gamma}{n}\left[\mathbb{E}_{P_{W}\otimes P_{S}}[ \ell(W,Z_i)]-\mathbb{E}_{P_{W,S}}[\ell(W,Z_i)]\right]\nonumber\\
        =&\ \gamma\left[\mathbb{E}_{P_{W}\otimes P_{S}}[ L_e(W,S)]-\mathbb{E}_{P_{W,S}}[L_e(W,S)]\right]\nonumber\\
        =&\ \gamma\cdot {\rm{gen}}(P^\gamma_{W|S},P_S).
    \end{align}

    Using the fact that for Gibbs algorithm, we have
    \begin{equation}\label{equ:I=gen}
        I_{\mathrm{SKL} }(W;S)=\gamma\cdot {\rm{gen}}(P^\gamma_{W|S},P_S).
    \end{equation}
    Combining \eqref{equ:gap=gen} and \eqref{equ:I=gen}, we get \eqref{eq:Th1_aim}, and that completes our proof.
\end{proof}

\subsection{Discussion on Jensen's gap and WCDG distribution}
For a given loss function $\ell(\theta,s)$ and a given measure of reference $P_0$, the WCDG distribution $P_{\hat{S}|\Theta=\theta}^{(P_0,\beta)}$ is a measure defined as
\begin{equation}
    \frac{dP_{\hat{S}|\Theta=\theta}^{(P_0,\beta)}}{dP_0}=\exp\bigg(\frac{1}{\beta}\ell(\theta,s)-\log\int\exp\bigg(\frac{1}{\beta}\ell(\theta,s)\bigg)dP_0(s)\bigg).
\end{equation}
Note that this $\ell(\theta,s)$ could be different from what we introduced in the preliminaries section. Therefore, each $\beta$ corresponds to a WCDG distribution, and if we assign
\begin{equation}
    \alpha = D(P_{\hat{S}|\Theta=\theta}^{(P_0,\beta)}\Vert P_0),
\end{equation}
then it can be proved that $P_{\hat{S}|\Theta=\theta}^{(P_0,\beta)}$ is the solution to the following optimization problem
\begin{equation}
    \begin{aligned}
        \max_{P\ll P_0}\int\ell(\theta,s)dP(s)\ \ 
        s.t.\ \ D(P\Vert P_0)\leq \alpha
    \end{aligned}
\end{equation}
where $\theta$, $P_0$ is fixed. In other words, $P_{\hat{S}|\Theta=\theta}^{(P_0,\beta)}$ represents the worst case distribution such that the expected loss is the highest while the distribution itself is not ``too far'' from a reference distribution $P_0$. \cite{zou2024worst} used calculus of variation with constraint to generate the result for this optimization problem by seeing it as a functional of
$dP/dP_0$, which one-to-one corresponds to a measure $P\ s.t.\ P\ll P_0$.

Under this setting, and we assign $\theta=w,z_i$, and $s=z^{-i}$, let
\begin{equation}\label{equ: WCDG loss}
    \ell(\theta,s)=L_e(w,z_i,z^{-i})-\frac{1}{\gamma}\log V_{L_e}(z_i,z^{-i},\gamma), 
\end{equation}
then we have
\begin{equation}
\mathbb{E}_{P_{W}\otimes P_{Z_i}}[J_i^{[n]}(W,Z_i)]=D(P_W\otimes P_S\Vert P_{\hat{Z}^{n-1},Z_i,W}^{(\mu^{n-1},\frac{1}{\gamma})})
\end{equation}
where $P_{\hat{Z}^{n-1},Z_i,W}^{(\mu^{n-1},\frac{1}{\gamma})}=P_{\hat{Z}^{n-1}|Z_i,W}^{(\mu^{n-1},\frac{1}{\gamma})}P_{Z_i,W}$, and $P_{\hat{Z}^{n-1}|Z_i,W}^{(\mu^{n-1},\frac{1}{\gamma})}$ indicates the worst case data generating distribution of $Z^{-i}$ that maximize the loss $\ell(\theta,s)$ given in \eqref{equ: WCDG loss}, with a fixed hypothesis $W$ and one sample $Z_i$.

This also results in an upper bound for the generalization error
\begin{align}
    \sum_{i=1}^n\bigg(I_{\mathrm{SKL} }(W;Z_i)+D(P_W\otimes P_S\Vert P_{\hat{Z}^{n-1},Z_i,W}^{(\mu^{n-1},\frac{1}{\gamma})})\bigg)&\geq I_{\mathrm{SKL} }(W;S)\\&=\gamma {\rm{gen}}(P^\gamma_{W|S},P_S).
\end{align}
Note that this bound is not order-wise optimal as it omitted another non-negative Jensen term.
\section{Asymptotic Results}\label{appendixB}
\subsection{Proof of Lemma 1}\label{:Lem1}
\begin{proof}
    It is easy to show that
    \begin{align}
        \lim_{n\to\infty}\bigg(\frac{dP_{W,Z^n}^{[n]}}{dP_{W}^{\infty}\otimes P_{Z^\infty}}\bigg)&=\lim_{n\to\infty}\bigg(\frac{dP_{W,Z^n}^{[n]}}{dP_{W}^{\infty}\otimes \mu^n}\bigg)\nonumber\\
        &=\lim_{n\to\infty}\frac{\pi(W)e^{-\gamma \frac{1}{n}\sum_{1}^n\ell(W,Z_i)}}{\int_{\mathcal{W}}\pi(w)e^{-\gamma \frac{1}{n}\sum_{1}^n\ell(w,Z_i)}dw}\cdot \frac{\int_{\mathcal{W}}\pi(w)e^{-\gamma L_\mu(w)}dw}{\pi(W)e^{-\gamma L_\mu(W)}}\nonumber\\
        &=\frac{\pi(W)\displaystyle\lim_{n\to\infty}e^{-\gamma \frac{1}{n}\sum_{1}^n\ell(W,Z_i)}}{\int_{\mathcal{W}}\pi(w)\displaystyle\lim_{n\to\infty}e^{-\gamma \frac{1}{n}\sum_{1}^n\ell(w,Z_i)}dw}\cdot \frac{\int_{\mathcal{W}}\pi(w)e^{-\gamma L_\mu(w)}dw}{\pi(W)e^{-\gamma L_\mu(W)}}\nonumber\\
        &=\frac{\pi(W)e^{-\gamma \lim_{n\to\infty}\frac{1}{n}\sum_{1}^n\ell(W,Z_i)}}{\int_{\mathcal{W}}\pi(w)e^{-\gamma \lim_{n\to\infty}\frac{1}{n}\sum_{1}^n\ell(w,Z_i)}dw}\cdot \frac{\int_{\mathcal{W}}\pi(w)e^{-\gamma L_\mu(w)}dw}{\pi(W)e^{-\gamma L_\mu(W)}}\nonumber\\
        &=\frac{\pi(W)e^{-\gamma L_\mu(W)}}{\int_{\mathcal{W}}\pi(w)e^{-\gamma L_\mu(w)}dw}\cdot \frac{\int_{\mathcal{W}}\pi(w)e^{-\gamma L_\mu(w)}dw}{\pi(W)e^{-\gamma L_\mu(W)}}\nonumber
        \\&=1\ \ a.s.
    \end{align}
    where we used the strong law of large numbers to show that for any fixed $W$,
    \begin{align}
        \lim_{n\to\infty}\frac{1}{n}\sum_{1}^n\ell(W,Z_i)=L_\mu(W)
    \end{align}
    almost surely with regard to $P_{Z^\infty}$, and thus holds almost surely in $P_{W}^{\infty}\otimes P_{Z^\infty}$ using Fubini's theorem.
    We also used the dominated convergence theorem to exchange the integral and limit.
\end{proof}
\subsection{Proof of Lemma 2}
\begin{proof}
    Using similar techniques as in the proof of the previous Lemma \ref{:Lem1}
    \begin{align}
        \varliminf_{n\to\infty}\bigg(\frac{dP_{W,Z_i}^{[n]}}{dP_{W}^{\infty}\otimes P_{Z^\infty}}\bigg)&=\varliminf_{n\to\infty}\bigg(\frac{P_{W|Z_i}^{[n]}}{P_{W}^{\infty}}\bigg)\nonumber\\
        &=\varliminf_{n\to\infty}\bigg(\int_{\mathcal{Z}^{n-1}}P_{W|Z_i,z^{-i}}^{[n]}d\mu^{n-1}(z^{-i})\bigg)\cdot\frac{d{\rm{Leb}}}{dP_{W}^{\infty}}\nonumber\\
        &\geq \bigg(\int_{\rm{Z^{\infty}}}\varliminf_{n\to\infty}P_{W|Z_i,z^{-i}}^{[n]}dP_{Z^\infty}(z^{\infty})\bigg)\cdot\frac{d{\rm{Leb}}}{dP_{W}^{\infty}}\nonumber\\
        &=\frac{\pi(W)e^{-\gamma L_\mu(W)}}{\int_{\mathcal{W}}\pi(w)e^{-\gamma L_\mu(w)}\cdot dw}\cdot \frac{\int_{\mathcal{W}}\pi(w)e^{-\gamma L_\mu(w)}\cdot dw}{\pi(W)e^{-\gamma L_\mu(W)}}\nonumber\\&=1
    \end{align}
    where the inequality is obtained using Fatou's lemma. Here, we use ${\rm Leb}$ to denote the Lebesgue measure on $\mathcal{W}\subset \mathbb{R}^d$. By saying that, we acknowledge that $\pi(w)$ is a distribution with regard to Lebesgue measure on $\mathcal{W}$. Of course, the proof still works for other measures even when Lebesgue measure doesn't exist on $\mathcal{W}$, namely, when $\mathcal{W}$ is an abstract space (for example a general function space) with a measure.
    
    From the Radon-Nikodym theorem, $dP_{W,Z_i}^{[n]}/{dP_{W}^{\infty}\otimes P_{Z^\infty}}$ is a measurable function on $\mathcal{F}_{W,Z^n}^{[n]}$, and thus measurable on $\mathcal{F}^\infty$. Consequently, $\displaystyle\varliminf_{n\to\infty}\left({dP_{W,Z_i}^{[n]}}/{dP_{W}^{\infty}\otimes P_{Z^\infty}}\right)$ is also measurable on $\mathcal{F}^\infty$. We consider the set $\Omega_\delta$ as 
    \begin{align}
    \Omega_\delta=\left\{\omega\in \mathcal{W}\times \mathcal{Z}^\infty\bigg\vert \displaystyle\varliminf_{n\to\infty}\left(\frac{dP_{W,Z_i}^{[n]}}{dP_{W}^{\infty}\otimes P_{Z^\infty}}\right)(\omega)>1\right\}, 
    \end{align}
    and the set is obviously measurable. Suppose that $P_{W}^{\infty}\otimes P_{Z^\infty}(\Omega_\delta)>0$, then
    \begin{align}
        1&=\varliminf_{n\to\infty}\int_{\Omega^\infty} \bigg(\frac{dP_{W,Z_i}^{[n]}}{dP_{W}^{\infty}\otimes P_{Z^\infty}}\bigg)dP_{W}^{\infty}\otimes P_{Z^\infty}\\
        &\geq\int_{\Omega^\infty}\varliminf_{n\to\infty}\bigg(\frac{dP_{W,Z_i}^{[n]}}{dP_{W}^{\infty}\otimes P_{Z^\infty}}\bigg)dP_{W}^{\infty}\otimes P_{Z^\infty}>1
    \end{align}
    generates contradiction, indicating that $P_{W}^{\infty}\otimes P_{Z^\infty}(\Omega_\delta)=0$, which proves the lemma.
\end{proof}
\subsection{Proof of Corollary 1}
\begin{proof}
This is a direct result of the dominated convergence theorem.
\end{proof}
\subsection{Proof of Lemma 3}
\begin{proof} First, we have
        \begin{align}\label{equ:lem3_eq1}
            n\cdot\bigg(1-\frac{dP_W^{[n]}\otimes P_{Z_i}}{dP_{W,Z_i}^{[n]}}\bigg)&=-\frac{n\cdot\int_{Z^n}(P_{W|z^n}-P_{W|Z_i,z^{-i}})d\mu^n(z^n)}{\int_{Z^n}(P_{W|Z_i,z^{-i}})d\mu^n(z^n)}\nonumber\\
            &=-\frac{\int_{Z^\infty} n(P_{W|z^n}-P_{W|Z_i,z^{-i}})dP_{Z^\infty}(z^{\infty})}{\int_{Z^\infty}(P_{W|Z_i,z^{-i}})dP_{Z^\infty}(z^{\infty})}.
        \end{align}
    For $\frac{f(t,x)}{(Af)(x)}$ where $A: g(t)\mapsto y$ is a linear functional. We expand it into
    \begin{align}\label{equ:general_expand_form}
        \frac{f(t,x)}{(Af)(x)}=&(f(t,x_0)+f_x(t,x_0)(x-x_0)+\frac{1}{2}f_{xx}(t,\xi)(x-x_0)^2)\cdot \bigg(\frac{1}{(Af(t,x_0))}\nonumber\\&-\frac{1}{(Af(t,x_0))^2}A(f_x(t,x_0)(x-x_0)+\frac{1}{2}f_{xx}(t,\xi)(x-x_0)^2)\nonumber\\&+\frac{1}{\xi_1^3}\big(A(f_x(t,x_0)(x-x_0)+\frac{1}{2}f_{xx}(t,\xi)(x-x_0)^2)\big)^2\bigg).
    \end{align}
    where we used the Lagrange remainder, and $\xi\in(x_0,x)$, $\xi_1\in(Af(t,x_0),Af(t,x))$. Applying this technique where we choose the linear operator as
    \begin{equation}
        A:f(w)\mapsto\int_{\mathcal{W}}\pi(w)f(w)dw.
    \end{equation}
    We have
            \begin{align}\label{equ:n Delta P}
                &\ n(P_{W|z^n}-P_{W|Z_i,z^{-i}})\nonumber\\
                =&\ \frac{\pi(W)\big(e^{-\gamma\frac{1}{n}\sum_{j=1}^{n}\ell(W,z_j)}-e^{-\gamma\frac{1}{n}\sum_{j=1,j\neq i}^{n}\ell(W,z_j)-\gamma\frac{1}{n}\ell(W,Z_i)}\big)}{\int_{\mathcal{W}}\pi(w)e^{-\gamma\frac{1}{n}\sum_{j=1}^{n}\ell(w,z_j)}\cdot dw-\int_{\mathcal{W}}\pi(w)\big(e^{-\gamma\frac{1}{n}\sum_{j=1}^{n}\ell(W,z_j)}-e^{-\gamma\frac{1}{n}\sum_{j=1,j\neq i}^{n}\ell(w,z_j)-\gamma\frac{1}{n}\ell(w,Z_i)}\big)dw}
                \nonumber\\=&\ \bigg\{\frac{\pi(W)e^{-\gamma\frac{1}{n}\sum_{j=1}^{n}\ell(W,z_j)}(\frac{\gamma}{n}(\ell(W,Z_i)-\ell(W,z))\cdot n)}{\int_{\mathcal{W}}\pi(w)e^{-\gamma\frac{1}{n}\sum_{j=1}^{n}\ell(w,z_j)}\cdot dw}\nonumber\\&-\frac{\pi(W)e^{-\gamma\frac{1}{n}\sum_{j=1}^{n}\ell(W,z_j)}}{[\int_{\mathcal{W}}\pi(w)e^{-\gamma\frac{1}{n}\sum_{j=1}^{n}\ell(w,z_j)}\cdot dw]^2}\int_{\mathcal{W}}\pi(w)e^{-\gamma\frac{1}{n}\sum_{j=1}^{n}\ell(w,z_j)}(\frac{\gamma}{n}(\ell(w,Z_i)-\ell(w,z))\cdot n)\cdot dw \bigg\}\nonumber\\
                &+\bigg[\frac{\pi(W)e^{-\gamma\frac{1}{n}\sum_{j=1}^{n}\ell(W,z_j)}}{\int_{\mathcal{W}}\pi(w)e^{-\gamma\frac{1}{n}\sum_{j=1}^{n}\ell(w,z_j)}\cdot dw}\bigg]R(W,Z_i,z^{-i},n^{-1})\nonumber\\
                =&\ \bigg[\frac{\pi(W)e^{-\gamma\frac{1}{n}\sum_{j=1}^{n}\ell(W,z_j)}}{\int_{\mathcal{W}}\pi(w)e^{-\gamma\frac{1}{n}\sum_{j=1}^{n}\ell(w,z_j)}\cdot dw}\bigg]\bigg(\gamma(\ell(W,Z_i)-\ell(W,z))-\mathbb{E}_{e}^{[n]}[\ell(W,Z_i)-\ell(W,z)]\nonumber\\
                &+R(W,Z_i,z^{-i},n^{-1})\bigg)\\
            \end{align}
        where $R(W,Z_i,z^{-i},n^{-1})$ is the remainder of the Taylor's expansion. It is easy to see that
        \begin{equation}
            \begin{aligned}
                R(W,Z_i,z^{-i},n^{-1})= \Theta\bigg(\frac{1}{n}\bigg).
            \end{aligned}
        \end{equation}
        Using the fact that the energy function is bounded, it is not difficult to show that
        $R(W,Z_i,z^{-i},n^{-1})$ is uniformly small. In other words, $n\cdot R(W,Z_i,z^{-i},n^{-1})$ is uniformly bounded, satisfying $\forall n\in\mathbb{N}, \forall W, \forall Z_i, \forall z^{-i}$, there exists a constant $C$, such that
        \begin{equation}\label{equ:lem3_bound_R}
            \begin{aligned}
                R(W,Z_i,z^{-i},n^{-1})\leq \frac{C}{n}.
            \end{aligned}
        \end{equation}
        Then using that the energy function is bounded, together with \eqref{equ:n Delta P} and \eqref{equ:lem3_bound_R}, we get that $\frac{n(P_{W|z^n}-P_{W|Z_i,z^{-i}})}{\pi(W)}$ is uniformly bounded, satisfying $\forall n\in\mathbb{N}, \forall W, \forall Z_i, \forall z^{-i}$, there exists a constant $C_1$, such that
        \begin{equation}\label{equ:lem3_bound1}
            \frac{n(P_{W|z^n}-P_{W|Z_i,z^{-i}})}{\pi(W)}\leq C_1.
        \end{equation}
        Note that here $\pi(W)$ may not be bounded, but it is not essential since $\pi(W)$ is a probability distribution.
        Afterwards, we will prove that the denominator in \eqref{equ:lem3_eq1} is lower-bounded with any given $W,Z_i$.
            \begin{align}\label{equ:lem3_bound2}
                \bigg|\int_{Z^\infty}(P_{W|Z_i,z^{-i}})dP_{Z^\infty}(z^{\infty})\bigg|&=\int_{Z^\infty}(P_{W|Z_i,z^{-i}})dP_{Z^\infty}(z^{\infty})\nonumber\\
                &= \int_{Z^\infty}\frac{\pi(W)e^{-\gamma \frac{1}{n}\sum_{1}^n\ell(W,Z_i)}}{\int_{\mathcal{W}}\pi(w)e^{-\gamma \frac{1}{n}\sum_{1}^n\ell(w,z_j)}dw}dP_{Z^\infty}(z^{\infty})\nonumber\\
                &\geq \frac{\pi(W)e^{-\gamma \sup\ell(w,z)}}{e^{-\gamma \inf\ell(w,z)}}.
            \end{align}
        With \eqref{equ:lem3_bound1} \eqref{equ:lem3_bound2}, we can bound \eqref{equ:lem3_eq1}
            \begin{align}
                \bigg|n\cdot\bigg(1-\frac{dP_W^{[n]}\otimes P_{Z_i}}{dP_{W,Z_i}^{[n]}}\bigg)\bigg|
                &=\frac{\big|\int_{Z^\infty} n(P_{W|z^n}-P_{W|Z_i,z^{-i}})dP_{Z^\infty}(z^{\infty})\big|}{\big|\int_{Z^\infty}(P_{W|Z_i,z^{-i}})dP_{Z^\infty}(z^{\infty})\big|}\nonumber\\
                &\leq \frac{\pi(W)\int_{Z^\infty} \big|\frac{n(P_{W|z^n}-P_{W|Z_i,z^{-i}})}{\pi(W)}\big|dP_{Z^\infty}(z^{\infty})}{\frac{\pi(W)e^{-\gamma \sup\ell(w,z)}}{e^{-\gamma \inf\ell(w,z)}}}\nonumber\\
                &\leq \frac{e^{-\gamma \inf\ell(w,z)}\cdot C_1}{e^{-\gamma \sup\ell(w,z)}},
            \end{align}
        and that proves the uniformly bounded conclusion.
        
        With all the bounded condition, we can now apply the dominated convergence theorem and the strong law of large numbers to get \eqref{equ:lem3_aim}
            \begin{align}
                \lim_{n\to\infty}n\cdot\bigg(1-\frac{dP_W^{[n]}\otimes P_{Z_i}}{dP_{W,Z_i}^{[n]}}\bigg)&=-\frac{\displaystyle\lim_{n\to\infty}\int_{Z^\infty} n(P_{W|z^n}-P_{W|Z_i,z^{-i}})dP_{Z^\infty}(z^{\infty})}{\displaystyle\lim_{n\to\infty}\int_{Z^\infty}(P_{W|Z_i,z^{-i}})dP_{Z^\infty}(z^{\infty})}\nonumber\\
                &=-\frac{\displaystyle\int_{Z^\infty}\lim_{n\to\infty} n(P_{W|z^n}-P_{W|Z_i,z^{-i}})dP_{Z^\infty}(z^{\infty})}{\displaystyle\int_{Z^\infty}\lim_{n\to\infty}(P_{W|Z_i,z^{-i}})dP_{Z^\infty}(z^{\infty})}\nonumber\\
                &=\left[\frac{\pi(W)e^{-\gamma L_\mu(W)}}{\int_{\mathcal{W}}\pi(w)e^{-\gamma L_\mu(w)}\cdot dw}\right]\bigg(-(\gamma(\ell(W,Z_i)-L_\mu(W)))\nonumber\\&+\mathbb{E}_{W}^{\infty}(\gamma(\ell(W,Z_i)-L_\mu(W))) \bigg)\cdot \frac{1}{\left[\frac{\pi(W)e^{-\gamma L_\mu(W)}}{\int_{\mathcal{W}}\pi(w)e^{-\gamma L_\mu(w)}\cdot dw}\right]}\nonumber\\
                &=(-\gamma(\ell(W,Z_i)-L_\mu(W)))+\mathbb{E}_{W}^{\infty}(\gamma(\ell(W,Z_i)-L_\mu(W))).
            \end{align}
        And that finishes our proof.
\end{proof}
\subsection{Proof of Theorem 2}
\begin{proof}
    From the definition
    \begin{equation}
        I_{\mathrm{SKL} }(W;Z_i)=\int_{W,Z^\infty}\bigg(\frac{dP_W^{[n]}\otimes P_{Z_i}}{dP_{W,Z_i}^{[n]}}\log\frac{dP_W^{[n]}\otimes P_{Z_i}}{dP_{W,Z_i}^{[n]}}-\log\frac{dP_W^{[n]}\otimes P_{Z_i}}{dP_{W,Z_i}^{[n]}}\bigg)dP_{W,Z^n}^{[n]}.
    \end{equation}
    Let
    \begin{equation}
        x=1-\frac{dP_W^{[n]}\otimes P_{Z_i}}{dP_{W,Z_i}^{[n]}},
    \end{equation}
    then the term inside the integration becomes
    \begin{equation}
        (1-x)\log(1-x)-\log(1-x)=x^2+\frac{1}{2(1-\xi)^2}x^3,
    \end{equation}
    where $\xi\in (0,x)$.

    Using Lemma \ref{:Lem3}, $|x|\leq \frac{C_1}{n}$ for some constant $C_1$, then for sufficiently large $n$, say $n>2C_1$, we easily get
    \begin{equation}
        \frac{1}{(1-\xi)^2}<4.
    \end{equation}
    Consider that
        \begin{align}
            \big|n^2\big((1-x)\log(1-x)-\log(1-x)\big)\big|
            &=n^2\bigg|x^2+\frac{1}{2(1-\xi)^2}x^3\bigg|\nonumber\\
            &\leq (nx)^2+n^2\frac{1}{2(1-\xi)^2}|x|^3\nonumber\\
            &\leq (nx)^2+2|nx|^3\cdot\frac{1}{n}\nonumber\\
            &\leq C_1^2+2C_1^3\cdot\frac{1}{n}
        \end{align}
    is uniformly bounded for sufficiently large $n$. From Corollary \ref{:Cor1} we know $\frac{dP_{W,Z_i}^{[n]}}{dP_{W}^{\infty}\otimes P_{Z^\infty}}$ is also uniformly bounded, therefore we can exchange the order of integration and limits using the dominated convergence theorem, which yields
        \begin{align}
            \lim_{n\to\infty}n^2I_{\mathrm{SKL} }(W;Z_i)&=\int \lim_{n\to\infty}n^2\big((1-x)\log(1-x)-\log(1-x)\big)\frac{dP_{W,Z_i}^{[n]}}{dP_{W}^{\infty}\otimes P_{Z^\infty}}dP_{W}^{\infty}\otimes P_{Z^\infty}(w,z^\infty)\nonumber\\
            &=\int \lim_{n\to\infty}n^2\bigg(x^2+\frac{1}{2(1-\xi)^2}x^3\bigg)\frac{dP_{W,Z_i}^{[n]}}{dP_{W}^{\infty}\otimes P_{Z^\infty}}dP_{W}^{\infty}\otimes P_{Z^\infty}(w,z^\infty)\nonumber\\
            &=\int \lim_{n\to\infty}n^2x^2\frac{dP_{W,Z_i}^{[n]}}{dP_{W}^{\infty}\otimes P_{Z^\infty}}dP_{W}^{\infty}\otimes P_{Z^\infty}(w,z^\infty)\nonumber\\
            &=\int \bigg((\gamma(\ell(w,z_i)-L_\mu(w)))-\mathbb{E}_{W}^{\infty}[\gamma(\ell(W,z_i)-{L_\mu(W)})]\bigg)^2 dP_{W}^{\infty}\otimes P_{Z^\infty}(w,z^\infty)\label{last_equality}\\
            &<+\infty
        \end{align}
    where the last equality \eqref{last_equality} is obtained using Corollary \ref{:Cor1} and Lemma \ref{:Lem3}. Therefore,
        \begin{align}\label{equ: n2Iskl}
            \lim_{n\to\infty}n^2\cdot I_{\mathrm{SKL} }(W;Z_i)&=\mathbb{E}_{P_{W}^{\infty}\otimes \mu}\bigg[\bigg((\gamma(\ell(W,Z)-L_\mu(W)))-\mathbb{E}_{W}^{\infty}[\gamma(\ell(W,Z)-{L_\mu(W)})]\bigg)^2\bigg]\\
            &=\gamma^2\mathbb{E}_{\mu}\bigg[\mathbb{E}_{W}^{\infty}[(\ell(W,Z)-L_\mu(W))^2]-\mathbb{E}_{W}^{\infty}[(\ell(W,Z)-L_\mu(W))]^2\bigg].
        \end{align}
\end{proof}
\subsection{Proof of Corollary 3}
\begin{proof}
    This conclusion is directly obtained from $I(W;Z_i)\leq I_{\mathrm{SKL} }(W;Z_i)$.
\end{proof}
\subsection{Proof of Remark 2}
\begin{proof}
    Using Lemma \ref{:Lem3}, we start by showing
    \begin{equation}
    \begin{aligned}
        &\lim_{n\to\infty}n\log\frac{dP_{W,Z_i}^{[n]}}{dP_W^{[n]}\otimes P_{Z_i}}=\bigg(-(\gamma(\ell(W,Z_i)-L_\mu(W)))+\mathbb{E}_{W}^{\infty}[\gamma(\ell(W,Z_i)-{L_\mu(W)})]\bigg).
    \end{aligned}
    \end{equation}
    Using similar techniques as in Theorem \ref{:Thm2}, we again have the terms inside the integration being uniformly bounded for sufficiently large $n$. Therefore, using the dominated convergence theorem, we have
    \begin{align}
    \lim_{n\to\infty}nI(W;Z_i)&=\int_{W,Z^\infty}\lim_{n\to\infty}n\log\bigg(\frac{dP_{W,Z_i}^{[n]}}{dP_W^{[n]}\otimes P_{Z_i}}\bigg)\cdot
        \bigg(\frac{dP_{W,Z_i}^{[n]}}{dP_{W}^{\infty}\otimes P_{Z^\infty}}\bigg)dP_{W}^{\infty}\otimes P_{Z^\infty}(w,z^\infty)\nonumber\\
        &=-\int_{W,Z_i}\bigg((\gamma(\ell(w,z_i)-L_\mu(w)))-\mathbb{E}_{W}^{\infty}[\gamma(\ell(W,z_i)-{L_\mu(W)})]\bigg)dP_{W}^{\infty}\otimes P_{Z_i}(w,z_i)\nonumber\\
        &=0,
    \end{align}
    which completes our proof.
\end{proof}
\subsection{Proof of Lemma 4}
\begin{proof} Without the lost of generality, we assume $\pi(w)$ is also bounded to simplify our proof. Note that this assumption can surely be taken off without affecting the conclusion.

    We first define \begin{align}
    \Delta_{z_i}(w)&\triangleq P_{W|Z^n}(w|z_i,z^{-i})-P_{W|Z^{n}}(w|z^{n})\\
    &=P_{w|z_i,z^{-i}}-P_{w|z^n}.
    \end{align}
    And therefore as proved in Lemma \ref{:Lem3}
        \begin{align}\label{equ:lem4_Delta}
        &\Delta_{z_i}(w,z^n)=\left[\frac{\pi(w)e^{-\gamma \frac{1}{n}\sum_{j=1}^n {\ell(w,z_j)}}}{\int_{\mathcal{W}}\pi(w)e^{-\gamma \frac{1}{n}\sum_{j=1}^n {\ell(w,z_j)}}\cdot dw}\right]\cdot\bigg(-(\gamma(\ell(w,z_i)-\ell(w,z)))+\mathbb{E}_{e}^{[n]}[\gamma(\ell(W,z_i)\nonumber\\&- {\ell(W,z)})]\bigg)\cdot\frac{1}{n}+o\bigg(\frac{1}{n}\bigg).
        \end{align}
    Note that if $\ell$ is only lower-bounded, the expansion still rings true almost surely. Only that this time, the $o\big(\frac{1}{n}\big)$ will no longer be uniformly small.
    
    Then we consider
        \begin{align}\label{equ:lem4_expand_nJ}
            n\cdot J_i^{[n]}(w,z_i)&=n\cdot\left[\log\bigg(\int_{\mathcal{Z}^{n-1}}P_{w|z_i,z^{-i}}d\mu^{n-1}(z^{-i}) \right)- \int_{\mathcal{Z}^{n-1}}\log(P_{w|z_i,z^{-i}})d\mu^{n-1}(z^{-i})\bigg]\nonumber\\
            &=n\cdot\left[\log\bigg(\int_{\rm{Z}^{n}}P_{w|z^{n}}+\Delta_{z_i}(w,z^n)\ d\mu^n(z^n) \right)\nonumber\\&- \int_{\rm{Z}^{n}}\log(P_{w|z^{n}}+\Delta_{z_i}(w,z^n))d\mu^n(z^n)\bigg].
        \end{align}
    For any $r\in [r_1,r_2]\subset\mathbb{R}$ where $r_1>0,r_2<+\infty$, let $0<m_1\leq r+x\leq m_2<+\infty$, it's easy to prove that there exists a real number $a<0$, such that
    \begin{equation}\label{ineq:lem4_log}
        \begin{aligned}
        \log(r+x)&\leq \log(r) + \frac{x}{r},\\
        \log(r+x)&\geq \log(r) + \frac{x}{r} + ax^2.
        \end{aligned}
    \end{equation}
    Now, we assign 
    \begin{equation}\label{eq:lem4_J hat}
    \hat{J}^{[n]}(w)=\log\bigg(\int_{\rm{Z}^{n}}P_{w|z^{n}}d\mu^n(z^n)\bigg)-\int_{\rm{Z}^{n}}\log P_{w|z^{n}}d\mu^n(z^n).
    \end{equation}
    From the proof of Lemma \ref{:Lem3}, we know that $n\cdot \Delta_{z_i}(w,z^n)$ is uniformly bounded. It is also easy to see that $P_{w|z^n}\in [r_1,r_2]$ where $r_1>0,r_2<+\infty$, and $0<r_1\leq P_{w|z^n}+\Delta_{z_i}(w,z^n)\leq r_2<+\infty$. Combining \eqref{equ:lem4_expand_nJ} \eqref{ineq:lem4_log} and \eqref{eq:lem4_J hat} we get
        \begin{align}
            n\cdot(\hat{J}^{[n]}(w)-J_i^{[n]}(w,z_i))&\leq \bigg[\int_{\rm{Z}^{n}}\frac{n\Delta_{z_i}(w,z^n)}{P_{w|z^{n}}}d\mu^n(z^n)-\frac{\int_{\rm{Z}^{n}}n\Delta_{z_i}(w,z^n)d\mu^n(z^n)}{\int_{\rm{Z}^{n}}P_{w|z^{n}}d\mu^n(z^n)}\bigg]\nonumber\\
            &-\frac{a}{n}\bigg(\int_{\rm{Z}^n}n\cdot \Delta_{z_i}(w,z^n)d\mu^n(z^n)\bigg)^2,
        \end{align}
        \begin{align}
            n\cdot(\hat{J}^{[n]}(w)-J_i^{[n]}(w,z_i))&\geq \bigg[\int_{\rm{Z}^{n}}\frac{n\Delta_{z_i}(w,z^n)}{P_{w|z^{n}}}d\mu^n(z^n)-\frac{\int_{\rm{Z}^{n}}n\Delta_{z_i}(w,z^n)d\mu^n(z^n)}{\int_{\rm{Z}^{n}}P_{w|z^{n}}d\mu^n(z^n)}\bigg]\nonumber\\
            &+\frac{a}{n}\int_{\rm{Z}^n}\bigg(n\cdot \Delta_{z_i}(w,z^n) \bigg)^2d\mu^n(z^n),
        \end{align}
    from which we get $n\cdot(\hat{J}^{[n]}(w)-J_i^{[n]}(w,z_i))$ is uniformly bounded. Then we apply the squeeze theorem, and yield
    \begin{equation}
        \begin{aligned}
        \lim_{n\to\infty}n\cdot(\hat{J}^{[n]}(w)-J_i^{[n]}(w,z_i))=\lim_{n\to\infty}\bigg[\int_{\rm{Z}^{n}}\frac{n\Delta_{z_i}(w,z^n)}{P_{w|z^{n}}}d\mu^n(z^n)-\frac{\int_{\rm{Z}^{n}}n\Delta_{z_i}(w,z^n)d\mu^n(z^n)}{\int_{\rm{Z}^{n}}P_{w|z^{n}}d\mu^n(z^n)}\bigg].\\
        \end{aligned}
    \end{equation}
    Eventually combining the result that $n\cdot \Delta_{z_i}(w,z^n)$ is uniformly bounded and \eqref{equ:lem4_Delta}, we complete our proof by showing
        \begin{align}
        \lim_{n\to\infty}n\cdot(\hat{J}^{[n]}(w)-J_i^{[n]}(w,z_i))&=\bigg[\int_{\rm{Z}^{\infty}}\frac{\displaystyle\lim_{n\to\infty}n\Delta_{z_i}(w,z^n)}{\displaystyle\lim_{n\to\infty}P_{w|z^{n}}}dP_{Z^\infty}(z^{\infty})\nonumber\\&-\frac{\int_{\rm{Z}^{\infty}}\displaystyle\lim_{n\to\infty}n\Delta_{z_i}(w,z^n)dP_{Z^\infty}(z^{\infty})}{\int_{\rm{Z}^{\infty}}\displaystyle\lim_{n\to\infty}P_{w|z^{n}}dP_{Z^\infty}(z^{\infty})}\bigg]\\
        &=0
        \end{align}
    where we used the strong low of large numbers and the dominated convergence theorem.
\end{proof}
\subsection{Proof of Lemma 5}
\begin{proof} We define
    \begin{equation}
        \delta(w,z^n)\triangleq \frac{1}{n}\sum_{i=0}^n\ell(w,z_i)-L_\mu(w).
    \end{equation}
    We first prove that
    \begin{equation}\label{equ:lem5 cod1}
        \int_{Z^n}\mathbb{E}_{W}^{\infty}\big[\sqrt{n}\cdot\vert\delta(W,z^n)\vert\big]^2d\mu^n(z^n),
    \end{equation}
    \begin{equation}\label{equ:lem5 cod2}
        \int_{Z^n}\mathbb{E}_{W}^{\infty}\big[\sqrt{n}\cdot|\delta(W,z^n)|\big]^4d\mu^n(z^n)
    \end{equation}
    are both bounded. For \eqref{equ:lem5 cod1}, using Cauchy's inequality
    \begin{align}\label{ieq:2mmt}
        \int_{Z^n}\mathbb{E}_{W}^{\infty}\big[\sqrt{n}\cdot\vert\delta(W,z^n)\vert\big]^2d\mu^n(z^n)&\leq \int_{Z^n}\mathbb{E}_{W}^{\infty}\big[n\cdot\delta(W,z^n)^2\big]d\mu^n(z^n)\nonumber\\
        &=\mathbb{E}_{W}^{\infty}\big[\int_{Z^n}n\cdot\delta(W,z^n)^2d\mu^n(z^n)\big]\nonumber\\
        &=\mathbb{E}_{W}^{\infty}\big[\mathbb{E}_{\mu^n}[n\cdot\delta(W,Z^n)^2]\big]\nonumber\\
        &=\mathbb{E}_{W}^{\infty}\bigg[\mathbb{E}_{\mu^n}[\frac{1}{n}\cdot\sum_{i=1}^n(\ell(W,Z_i)-L_\mu(W))^2]+\nonumber\\
        &\begin{matrix}
        \underbrace{\mathbb{E}_{\mu^n}[\frac{1}{n}\sum_{i\neq j}(\ell(W,Z_i)-L_\mu(W))(\ell(W,Z_j)-L_\mu(W))]}\bigg] \\ =0
        \end{matrix}\nonumber\\
        &=\mathbb{E}_{W}^{\infty}\bigg[\mathbb{E}_{\mu}[(\ell(W,Z)-L_\mu(W))^2]\bigg]
    \end{align}
    and is therefore bounded.

    Similarly for \eqref{equ:lem5 cod2}, we again use Cauchy's inequality and yield
        \begin{align}\label{ieq:4mmt}
        \int_{Z^n}\mathbb{E}_{W}^{\infty}\big[\sqrt{n}\cdot\vert\delta(W,z^n)\vert\big]^4d\mu^n(z^n)&\leq \int_{Z^n}\mathbb{E}_{W}^{\infty}\big[n^2\cdot\delta(W,z^n)^4\big]d\mu^n(z^n)\nonumber\\
        &=\mathbb{E}_{W}^{\infty}\big[\int_{Z^n}n^2\cdot\delta(W,z^n)^4d\mu^n(z^n)\big]\nonumber\\
        &=\mathbb{E}_{W}^{\infty}\big[\mathbb{E}_{\mu^n}[n^2\cdot\delta(W,Z^n)^4]\big]\nonumber\\
        &=\mathbb{E}_{W}^{\infty}\bigg[\frac{1}{n}\mathbb{E}_{\mu}[(\ell(W,Z)-L_\mu(W))^4]+\nonumber\\&\frac{3(n-1)}{n}\mathbb{E}_{\mu}[(\ell(W,Z)-L_\mu(W))^2]^2\bigg]
    \end{align}
    and is therefore bounded as well. And as will be discussed in Annotation \ref{:Rmk2}, this result holds true for all power $k$ besides $2$ and $4$.
    
    Notice that $\hat{J}^{[n]}(w)$ can also be written as
    \begin{equation}\label{equ:lem5 J hat}
        \hat{J}^{[n]}(w)=\log\int_{\rm{Z}^{n}}\frac{P_{w|z^{n}}}{P_{w}^{\infty}}d\mu^n(z^n)-\int_{\rm{Z}^{n}}\log \frac{P_{w|z^{n}}}{P_{w}^{\infty}}d\mu^n(z^n).
    \end{equation}
    Using Taylor's theorem with Lagrange form remainder (similar to \eqref{equ:general_expand_form}), we have the following expansion
        \begin{align}\label{equ:lem5 expand}
            \frac{P_{w|z^{n}}}{P_{w}^{\infty}}&=\frac{\pi(w)e^{-\gamma (L_\mu(w)+\delta(w,z^n))}}{\int_{\mathcal{W}}\pi(w)e^{-\gamma (L_\mu(w)+\delta(w,z^n))}\cdot dw}\cdot\frac{\int_{\mathcal{W}}\pi(w)e^{-\gamma L_\mu(w)}\cdot dw}{\pi(w)e^{-\gamma L_\mu(w)}}\nonumber\\
            &=\bigg(1+R_1(w,\delta)\delta(w,z^n)\bigg)\cdot\bigg(1-R_2(w,\delta)\cdot\mathbb{E}_{W}^{\infty}[R_1(W,\delta)\delta(W,z^n)]\bigg)
        \end{align}
    where $R_1(w,\delta),R_2(w,\delta)$ are both bounded. More specifically
    \begin{equation}
        R_1(w,\delta)=-\gamma e^{-\gamma\xi_1}\ \ \xi_1\in(0,\delta),
    \end{equation}
    \begin{equation}
        R_2(w,\delta)=\bigg(\frac{\int_{\mathcal{W}}\pi(w)e^{-\gamma L_\mu(w)}\cdot dw}{\xi_2}\bigg)^2,
    \end{equation}
    where $\xi_2\in(\int_{\mathcal{W}}\pi(w)e^{-\gamma L_\mu(w)}\cdot dw, \int_{\mathcal{W}}\pi(w)e^{-\gamma (L_\mu(w)+\delta(w,z^n))}\cdot dw)$. For simplicity, we respectively use $M_1$ and $M_2$ to each represent a upper-bound for $|R_1|$ and $|R_2|$.
    
    We employ the same technique used in the proof of Lemma \ref{:Lem4} to bound $\log(r+x)$, where $r=1$ in this circumstance. With this, we are able to dominate $n\cdot \hat{J}^{[n]}(w)$ using the expanded form of $\delta(w,z^n)$. This time for some $a<0$
        \begin{align}\label{ineq:lem5 bound n J}
            n\cdot\hat{J}^{[n]}(w)\leq& -a\int_{\mathcal{Z}^n}n\cdot\bigg(R_1(w,\delta)\delta(w,z^n)-R_2(w,\delta)\mathbb{E}_{W}^{\infty}[R_1(W,\delta)\delta(W,z^n)]\nonumber\\&
            -R_1(w,\delta)R_2(w,\delta)\delta(w,z^n)\mathbb{E}_{W}^{\infty}[R_1(W,\delta)\delta(W,z^n)]\bigg)^2d\mu^n(z^n)\nonumber\\
            \leq& |a|\bigg\{
            M_1^2\int_{\mathcal{Z}^n}|\sqrt{n}\delta(w,z^n)|^2d\mu^n(z^n)+M_1^2M_2^2\int_{\mathcal{Z}^n}\mathbb{E}_{W}^{\infty}[\sqrt{n}\cdot|\delta(w,z^n)|]^2d\mu^n(z^n)\nonumber\\
            &+2M_1^2M_2\int_{\mathcal{Z}^n}|\sqrt{n}\delta(w,z^n)|\cdot\mathbb{E}_{W}^{\infty}[\sqrt{n}\cdot|\delta(W,z^n)|]d\mu^n(z^n)\nonumber\\
            &+\frac{1}{\sqrt{n}}\cdot 2M_1^3M_2\int_{\mathcal{Z}^n}|\sqrt{n}\delta(w,z^n)|^2\cdot\mathbb{E}_{W}^{\infty}[\sqrt{n}\cdot|\delta(W,z^n)|]d\mu^n(z^n)\nonumber\\
            &+\frac{1}{\sqrt{n}}\cdot 2M_1^3M_2^2\int_{\mathcal{Z}^n}|\sqrt{n}\delta(w,z^n)|\cdot\mathbb{E}_{W}^{\infty}[\sqrt{n}\cdot|\delta(W,z^n)|]^2d\mu^n(z^n)\nonumber\\
            &+\frac{1}{n}\cdot M_1^4M_2^2\int_{\mathcal{Z}^n}|\sqrt{n}\delta(w,z^n)|^2\cdot\mathbb{E}_{W}^{\infty}[\sqrt{n}\cdot|\delta(W,z^n)|]^2d\mu^n(z^n)
            \bigg\}\nonumber\\
            \leq & |a|\bigg\{
            M_1^2\int_{\mathcal{Z}^n}|\sqrt{n}\delta(w,z^n)|^2d\mu^n(z^n)+M_1^2M_2^2\int_{\mathcal{Z}^n}\mathbb{E}_{W}^{\infty}[\sqrt{n}\cdot|\delta(w,z^n)|]^2d\mu^n(z^n)\nonumber\\
            &+2M_1^2M_2\bigg(\int_{\mathcal{Z}^n}|\sqrt{n}\delta(w,z^n)|^2d\mu^n(z^n)\bigg)^{\frac{1}{2}}\bigg(\int_{\mathcal{Z}^n}\mathbb{E}_{W}^{\infty}[\sqrt{n}\cdot|\delta(W,z^n)|]^2d\mu^n(z^n)\bigg)^{\frac{1}{2}}\nonumber\\
            &+\frac{1}{\sqrt{n}}\cdot 2M_1^3M_2\bigg(\int_{\mathcal{Z}^n}|\sqrt{n}\delta(w,z^n)|^4d\mu^n(z^n)\bigg)^{\frac{1}{2}}\bigg(\int_{\mathcal{Z}^n}\mathbb{E}_{W}^{\infty}[\sqrt{n}\cdot|\delta(W,z^n)|]^2d\mu^n(z^n)\bigg)^{\frac{1}{2}}\nonumber\\
            &+\frac{1}{\sqrt{n}}\cdot 2M_1^3M_2^2\bigg(\int_{\mathcal{Z}^n}|\sqrt{n}\delta(w,z^n)|^2d\mu^n(z^n)\bigg)^{\frac{1}{2}}\bigg(\int_{\mathcal{Z}^n}\mathbb{E}_{W}^{\infty}[\sqrt{n}\cdot|\delta(W,z^n)|]^4d\mu^n(z^n)\bigg)^{\frac{1}{2}}\nonumber\\
            &+\frac{1}{n}\cdot M_1^4M_2^2\bigg(\int_{\mathcal{Z}^n}|\sqrt{n}\delta(w,z^n)|^4d\mu^n(z^n)\bigg)^{\frac{1}{2}}\bigg(\int_{\mathcal{Z}^n}\mathbb{E}_{W}^{\infty}[\sqrt{n}\cdot|\delta(W,z^n)|]^4d\mu^n(z^n)\bigg)^{\frac{1}{2}}
            \bigg\}.
        \end{align}
    Let
    \begin{equation}\label{equ:lem5 mmt2}
        {\rm{Mmt}_2}(\ell(w,Z))={\rm{Var}}(\ell(w,Z))=\int_{\rm{Z}}(\ell(w,z)-L_\mu(w))^2dP(z),
    \end{equation}
    \begin{equation}\label{equ:lem5 mmt4}
        {\rm{Mmt}_4}(\ell(w,Z))=\int_{\rm{Z}}(\ell(w,z)-L_\mu(w))^4dP(z),
    \end{equation}
    then it can be calculated that
    \begin{equation}\label{equ:lem5 var1}
        \int_{\mathcal{Z}^n}|\sqrt{n}\delta(w,z^n)|^2d\mu^n(z^n)={\rm{Var}}(\ell(w,Z)),
    \end{equation}
    \begin{equation}\label{equ:lem5 var2}
        \int_{\mathcal{Z}^n}|\sqrt{n}\delta(w,z^n)|^4d\mu^n(z^n)=\frac{1}{n}{\rm{Mmt}_4}(\ell(w,Z))+\frac{3(n-1)}{n}{\rm{Var}}(\ell(w,Z))^2.
    \end{equation}
    With the bounded energy function, \eqref{equ:lem5 mmt2} \eqref{equ:lem5 mmt4} \eqref{equ:lem5 var1} and \eqref{equ:lem5 var2} are all uniformly bounded, thus for some $C_1>0, C_2>0$,
    \begin{equation}
        \int_{\mathcal{Z}^n}|\sqrt{n}\delta(w,z^n)|^2d\mu^n(z^n)\leq C_1,
    \end{equation}
    \begin{equation}
        \int_{\mathcal{Z}^n}|\sqrt{n}\delta(w,z^n)|^4d\mu^n(z^n)\leq C_2.
    \end{equation}
    Then with condition \eqref{equ:lem5 cod1} and \eqref{equ:lem5 cod2}, we assign upper-bounds $C_3$ and $C_4$ for them respectively, getting
    \begin{equation}
        \int_{Z^n}\mathbb{E}_{W}^{\infty}\big[\sqrt{n}\cdot\vert\delta(W,z^n)\vert\big]^2d\mu^n(z^n)\leq C_3,
    \end{equation}
    \begin{equation}
        \int_{Z^n}\mathbb{E}_{W}^{\infty}\big[\sqrt{n}\cdot|\delta(W,z^n)|\big]^4d\mu^n(z^n)\leq C_4.
    \end{equation}
    Now, \eqref{ineq:lem5 bound n J} becomes
        \begin{align}
            n\cdot\hat{J}^{[n]}(w)\leq&\  |a|\bigg(M_1^2C_1+M_1^2M_2^2C_3+2M_1^2M_2C_1^{\frac{1}{2}}C_3^{\frac{1}{2}}
            +\frac{1}{\sqrt{n}}\cdot 2M_1^3M_2C_2^{\frac{1}{2}}C_3^{\frac{1}{2}}\nonumber\\&
            +\frac{1}{\sqrt{n}}\cdot 2M_1^3M_2^2C_1^{\frac{1}{2}}C_4^{\frac{1}{2}}
            +\frac{1}{n}\cdot M_1^4M_2^2C_2^{\frac{1}{2}}C_4^{\frac{1}{2}}\bigg),
        \end{align}
    and is therefore dominated by a constant, which surely proves that $n\cdot\hat{J}^{[n]}(w)$ is uniformly bounded. From this we also see that the terms with order $k\geq 3$ (in other words, containing $\hat{J}^{[n]}(w)^k$) are controlled by a factor $(1/n)^{(k/2-1)}$, thus those terms converges to zero with $n\to\infty$. This observation is used many times in later proofs.

    Then, we prove that $\lim_{n\to\infty}n\cdot\hat{J}^{[n]}(w)$ exists. Similar to \eqref{equ:lem5 expand}, we have
        \begin{align}\label{equ:lem5 new expand}
            \frac{P_{w|z^{n}}}{P_{w}^{\infty}}&=\frac{\pi(w)e^{-\gamma (L_\mu(w)+\delta(w,z^n))}}{\int_{\mathcal{W}}\pi(w)e^{-\gamma (L_\mu(w)+\delta(w,z^n))}\cdot dw}\cdot\frac{\int_{\mathcal{W}}\pi(w)e^{-\gamma L_\mu(w)}\cdot dw}{\pi(w)e^{-\gamma L_\mu(w)}}\nonumber\\
            &=\bigg(1-\gamma\delta(w,z^n)+R_1(w,\delta)\delta(w,z^n)^2\bigg)\cdot\bigg(1-\mathbb{E}_{W}^{\infty}[-\gamma\delta(W,z^n)+R_1(W,\delta)\delta(W,z^n)^2]\nonumber\\&+R_2(w,\delta)\cdot\mathbb{E}_{W}^{\infty}[-\gamma\delta(W,z^n)+R_1(W,\delta)\delta(W,z^n)^2]^2\bigg)
        \end{align}
    where $R_1(w,\delta)=\frac{1}{2}\gamma^2e^{-\gamma\xi_1}$, $\xi_1\in(0,\delta)$, $R_2(w,\delta)=\bigg(\frac{\int_{\mathcal{W}}\pi(w)e^{-\gamma L_\mu(w)}\cdot dw}{\xi_2}\bigg)^3$, and\\$\xi_2\in(\int_{\mathcal{W}}\pi(w)e^{-\gamma L_\mu(w)}\cdot dw, \int_{\mathcal{W}}\pi(w)e^{-\gamma (L_\mu(w)+\delta(w,z^n))}\cdot dw)$.
    
    Again similar to \eqref{ineq:lem4_log}, we have for some $a<0$
    \begin{equation}\label{ineq:lem5 new log1}
        \log(1+x)\leq x-\frac{1}{2}x^2+\frac{1}{3}x^3,
    \end{equation}
    \begin{equation}\label{ineq:lem5 new log2}
        \log(1+x)\geq x-\frac{1}{2}x^2+\frac{1}{3}x^3+ax^4,
    \end{equation}
    with $0<m\leq 1+x\leq M<+\infty$. Very similar to \eqref{ineq:lem5 bound n J}, combining \eqref{equ:lem5 J hat}, \eqref{equ:lem5 new expand}, \eqref{ineq:lem5 new log1} and \eqref{ineq:lem5 new log2}, we yield an upper-bound and a lower-bound of $n\cdot\hat{J}^{[n]}(w)$. It is easy to see that the terms with order equal to or greater than 3 converges to zero. Therefore, using the squeeze theorem, the upper-bound and lower-bound converges to the same value characterized by terms with order 1 and 2, and we eventually get
    \begin{align}
        \lim_{n\to\infty}n\cdot\hat{J}^{[n]}(w)&=\begin{matrix}\displaystyle
        \lim_{n\to\infty}\underbrace{-\frac{1}{2}\bigg[\int_{Z^n}-\gamma\sqrt{n}\delta(w,z^n)d\mu^n(z^n)+\int_{Z^n}\gamma\mathbb{E}_{W}^{\infty}[\sqrt{n}\delta(W,z^n)]d\mu^n(z^n)\bigg]^2}\\=0\end{matrix}\nonumber\\&
        +\int_{Z^n}\frac{1}{2}\bigg(-\gamma\sqrt{n}\delta(w,z^n)+\gamma\mathbb{E}_{W}^{\infty}[\sqrt{n}\delta(W,z^n)]\bigg)^2d\mu^n(z^n)\nonumber\\
        &=\lim_{n\to\infty}\int_{Z^n}\frac{1}{2}\bigg(-\gamma\sqrt{n}\delta(w,z^n)+\gamma\mathbb{E}_{W}^{\infty}[\sqrt{n}\delta(W,z^n)]\bigg)^2d\mu^n(z^n)\nonumber\\
        &=\frac{1}{2}\gamma^2\mathbb{E}_{\mu}[(\ell(w,Z)-L_\mu(w))^2]-\gamma^2\mathbb{E}_{\mu}\big[(\ell(w,Z)-L_\mu(w))\mathbb{E}_{W}^{\infty}[\ell(W,Z)-L_\mu(W)]\big]\nonumber\\&+\frac{1}{2}\gamma^2\mathbb{E}_{\mu}\big[\mathbb{E}_{W}^{\infty}[\ell(W,Z)-L_\mu(W)]^2\big],
    \end{align}
    which is a bounded function of $w$. That proves the lemma.
\end{proof}
\begin{annotation}\label{:Rmk2}
    According to \cite{47c0448d-fb64-36c8-80c1-f76c50302dbf}, for a series of i.i.d sampled $X$ with zero mean, when the $k^{\text{th}}$-moment of $X$ exists , 
    \begin{equation}\label{ieq:1962}
    \begin{aligned}
    \mathbb{E}\bigg[\bigg(\frac{|X_1+X_2+\cdots+X_n|}{\sqrt{n}}\bigg)^{k}\bigg]<C_k
    \end{aligned}
    \end{equation}
    for some constant $C_k$ depending only on the distribution of $X$. 
    
    Following its proof, it is easy to see that $C_k$ only depends on $\mathbb{E}[|X|^k]$ and $C_{k-1}$. The idea of the proof is to use mathematical induction. Knowing that when $k=2$, the conclusion holds for a $C_2=\mathbb{E}[X^2]$. Then if the conclusion holds for $k$ with $C_k$, it is able to prove that there exists a $C_{k+1}$ which only depends on $C_k$ and $\mathbb{E}[|X|^{k+1}]$ for the inequality to hold for all $n$.
    
    Therefore, for $\forall k>0$, the $k^{\text{th}}$-moment of $\ell(w,Z)$, $\mathbb{E}[|\ell(w,Z)|^{k}]$ surly exists, and is uniformly bounded since $\ell(w,z)$ is a bounded function. We then have the $C_k$ introduced above is uniform over $w$ for $\ell(w,Z)-L_\mu(w)$, in other words,
    \begin{equation}\label{ieq:kmomt1}
    \exists M_k<+\infty,\forall n,\forall w\ \mathbb{E}_{\mu^n}[(\sqrt{n}|\delta(w,Z^n)|)^{k}]<M_k.
    \end{equation}
    And similar to our proof in \eqref{ieq:2mmt} and \eqref{ieq:4mmt}, 
    \begin{equation}\label{ieq:kmomt2}
    \exists M_k<+\infty,\forall n,\ \mathbb{E}_{\mu^n}[(\sqrt{n}\mathbb{E}_{W}^{\infty}[|\delta(W,Z^n)|])^{k}]<M_k.
    \end{equation}
    \textbf{The result plays an important role in guaranteeing the third or higher order terms converges to zero.}

    For a more general setting, consider $2<p<\infty$ and independent random variables $X_i\in L^p$ with zero mean that are not necessarily identical. Then by \cite{rosenthal1970subspaces}, we have
    \begin{equation}
        \bigg(\mathbb{E}\bigg[\bigg|\sum_{i=1}^nX_i\bigg|^p\bigg]\bigg)^{\frac{1}{p}}\leq K_p\max\bigg\{\bigg(\sum_{i=1}^n\mathbb{E}[|X_i|^p]\bigg)^{\frac{1}{p}},\bigg(\sum_{i=1}^n\mathbb{E}[|X_i|^2]\bigg)^{\frac{1}{2}}\bigg\}
    \end{equation}
    and
    \begin{equation}
        \bigg(\mathbb{E}\bigg[\bigg|\sum_{i=1}^nX_i\bigg|^p\bigg]\bigg)^{\frac{1}{p}}\geq \frac{1}{2}\max\bigg\{\bigg(\sum_{i=1}^n\mathbb{E}[|X_i|^p]\bigg)^{\frac{1}{p}},\bigg(\sum_{i=1}^n\mathbb{E}[|X_i|^2]\bigg)^{\frac{1}{2}}\bigg\},
    \end{equation}
    where $K_p$ is a constant only depending on $p$. Note that this theorem can be seen as a version of Khintchine inequality. 
    
    It provide an estimate of the moment of the sum of independent random variables in a broader setting. Letting all the random variables be identically distributed, and the result recovers to \eqref{ieq:1962}, and can surely yield \eqref{ieq:kmomt1} and \eqref{ieq:kmomt2}.
\end{annotation}
\subsection{Proof of Theorem 3}
\begin{proof}
    We define
    \begin{equation}
        K_0(w)\triangleq \lim_{n\to\infty}n\cdot \hat{J}^{[n]}(w).
    \end{equation}
    The existence of the limit is obtained using Lemma \ref{:Lem5}. Using Lemma \ref{:Lem4} and Lemma \ref{:Lem5}, we get
    \begin{equation}
        \lim_{n\to\infty}n\cdot J_i^{[n]}(w,z_i)=K_0(w)
    \end{equation}
    and
        \begin{align}
            0\leq n\cdot J_i^{[n]}(w,z_i)&=n\cdot \hat{J}^{[n]}(w) - n\cdot(\hat{J}^{[n]}(w)-J_i^{[n]}(w,z_i))\\
            &\leq |n\cdot J_i^{[n]}(w,z_i)| +|n\cdot(\hat{J}^{[n]}(w)-J_i^{[n]}(w,z_i))|
        \end{align}
    which is uniformly bounded.

    Using Lemma \ref{:Lem3}, we know that $\displaystyle n\cdot\bigg(1-\frac{dP_W^{[n]}\otimes P_{Z_i}}{dP_{W,Z_i}^{[n]}}\bigg)$ is uniformly bounded. Therefore, together with Corollary \ref{:Cor1}, we have
        \begin{align}
            &\ \lim_{n\to\infty}n\cdot \bigg(\sum_{i=1}^nI_{\mathrm{SKL} }(W;Z_i)-I_{\mathrm{SKL} }(W;S)\bigg)\nonumber\\
            =&\ \lim_{n\to\infty}\int_{W,Z^\infty}n\cdot J_i^{[n]}(w,z_i)\cdot n\cdot\bigg(1-\frac{dP_W^{[n]}\otimes P_{Z_i}}{dP_{W,Z_i}^{[n]}}\bigg)\left(\frac{dP_{W,Z_i}^{[n]}}{dP_{W}^{\infty}\otimes P_{Z^\infty}}\right)dP_{W}^{\infty}\otimes P_{Z^\infty}(w,z^\infty) \nonumber\\
            =&\ \int_{W,Z_i}K_0(w)\bigg(-(\gamma(\ell(w,z_i)-L_\mu(w)))+\mathbb{E}_{W}^{\infty}[\gamma(\ell(W,z_i)-{L_\mu(W)})]\bigg)dP_{W}^{\infty}\otimes P_{Z_i}(w,z_i)\nonumber\\
            =&\ 0,
        \end{align}
    which proves the theorem.
\end{proof}
\subsection{Proof of Theorem 4}
\begin{proof}
    The proof of this theorem utilizes the very same techniques in the proof of Lemma \ref{:Lem5}. Similar to our previous expansion \eqref{equ:lem5 expand} and \eqref{equ:lem5 new expand}, we can write
        \begin{align}\label{equ:lem5 new new expand}
            \frac{P_{w|z^{n}}}{P_{w}^{\infty}}&=\frac{\pi(w)e^{-\gamma (L_\mu(w)+\delta(w,z^n))}}{\int_{\mathcal{W}}\pi(w)e^{-\gamma (L_\mu(w)+\delta(w,z^n))}\cdot dw}\cdot\frac{\int_{\mathcal{W}}\pi(w)e^{-\gamma L_\mu(w)}\cdot dw}{\pi(w)e^{-\gamma L_\mu(w)}}\nonumber\\
            &=\bigg(1-\gamma\delta(w,z^n)+\frac{1}{2}\gamma^2\delta(w,z^n)^2+R_1(w,\delta)\delta(w,z^n)^3\bigg)\cdot\bigg(1-\mathbb{E}_{W}^{\infty}[-\gamma\delta(W,z^n)+\frac{1}{2}\gamma^2\delta(W,z^n)^2\nonumber\\&+R_1(W,\delta)\delta(W,z^n)^3]+\mathbb{E}_{W}^{\infty}[-\gamma\delta(W,z^n)+\frac{1}{2}\gamma^2\delta(W,z^n)^2+R_1(W,\delta)\delta(W,z^n)^3]^2\nonumber\\&+R_2(w,\delta)\cdot\mathbb{E}_{W}^{\infty}[-\gamma\delta(W,z^n)+\frac{1}{2}\gamma^2\delta(W,z^n)^2+R_1(W,\delta)\delta(W,z^n)^3]^3\bigg).
        \end{align}
    We once again assign $x=1-\frac{P_{w|z^{n}}}{P_{w}^{\infty}}$, then similar to \eqref{ineq:lem5 new log1} and \eqref{ineq:lem5 new log2}, for some $a>0$,
    \begin{equation}\label{ineq:thm5 new new log1}
        -\log(1-x)\geq x+\frac{1}{2}x^2+\frac{1}{3}x^3,
    \end{equation}
    \begin{equation}\label{ineq:thm5 new new log2}
        -\log(1-x)\leq x+\frac{1}{2}x^2+\frac{1}{3}x^3+ax^4.
    \end{equation}
    Thus for lautum information, we can calculate the value for any given $n$ as
    \begin{equation}
        \begin{aligned}
            n\cdot L(W;S)=&n\cdot \int_{\mathcal{W},\mathcal{Z}^n}-\log(1-x)dP_{W}^{[n]}\otimes \mu^n(w,z^n).
        \end{aligned}
    \end{equation}
    Therefore,
    \begin{equation}\label{ineq:Thm5 squeeze1}
        \begin{aligned}
            n\cdot L(W;S)\leq n\cdot \int_{\mathcal{W},\mathcal{Z}^n}(x+\frac{1}{2}x^2+\frac{1}{3}x^3+ax^4)dP_{W}^{[n]}\otimes \mu^n(w,z^n),
        \end{aligned}
    \end{equation}
    \begin{equation}\label{ineq:Thm5 squeeze2}
        \begin{aligned}
            n\cdot L(W;S)\geq n\cdot \int_{\mathcal{W},\mathcal{Z}^n}(x+\frac{1}{2}x^2+\frac{1}{3}x^3)dP_{W}^{[n]}\otimes \mu^n(w,z^n),
        \end{aligned}
    \end{equation}
    
    It is natural to then combine \eqref{equ:lem5 new new expand} \eqref{ineq:Thm5 squeeze1} and \eqref{ineq:Thm5 squeeze2}. Once again, using the squeeze theorem, the terms with order higher or equal to 3 converges to zero (as proved in Lemma \ref{:Lem5} and Annotation \ref{:Rmk2}), and the two sides of the inequalities converges to the same value as shown below
        \begin{align}\label{equ:Thm5 nL}
            &\ \lim_{n\to\infty} n\cdot L(W;S)\nonumber\\=&\ \lim_{n\to\infty}n\cdot \int_{\mathcal{W},\mathcal{Z}^n}\bigg[\gamma\bigg(-\delta(w,z^n)+\mathbb{E}_{W}^{\infty}[\delta(W,z^n)]\bigg)+\nonumber\\
            &\ \gamma^2\bigg(-\delta(w,z^n)\mathbb{E}_{W}^{\infty}[\delta(W,z^n)]+\frac{1}{2}\delta(w,z^n)^2-\frac{1}{2}\mathbb{E}_{W}^{\infty}[\delta(W,z^n)^2]+\mathbb{E}_{W}^{\infty}[\delta(W,z^n)]^2+\nonumber\\
            &\ \frac{1}{2}\big(\mathbb{E}_{W}^{\infty}[\delta(W,z^n)]-\delta(w,z^n)\big)^2
            \bigg)\bigg]dP_{W}^{[n]}\otimes \mu^n(w,z^n)\nonumber\\
            =&\ \lim_{n\to\infty}n\cdot \int_{\mathcal{W},\mathcal{Z}^n}\gamma^2\bigg(-\delta(w,z^n)\mathbb{E}_{W}^{\infty}[\delta(W,z^n)]+\frac{1}{2}\delta(w,z^n)^2-\frac{1}{2}\mathbb{E}_{W}^{\infty}[\delta(W,z^n)^2]+\nonumber\\
            &\ \mathbb{E}_{W}^{\infty}[\delta(W,z^n)]^2+\frac{1}{2}\big(\mathbb{E}_{W}^{\infty}[\delta(W,z^n)]-\delta(w,z^n)\big)^2
            \bigg)dP_{W}^{[n]}\otimes \mu^n(w,z^n)\nonumber\\
            =&\ \lim_{n\to\infty}\gamma^2n\cdot \int_{\mathcal{W},\mathcal{Z}^n}\big(\mathbb{E}_{W}^{\infty}[\delta(W,z^n)]-\delta(w,z^n)\big)^2-\frac{1}{2}\big(\mathbb{E}_{W}^{\infty}[\delta(W,z^n)^2]-\mathbb{E}_{W}^{\infty}[\delta(W,z^n)]^2\big)dP_{W}^{[n]}\otimes \mu^n(w,z^n)\nonumber\\
            =&\ \gamma^2\lim_{n\to\infty}\bigg\{\mathbb{E}_{\mu}\big[\mathbb{E}_{W}^{\infty}[(\ell(W,Z)-L_\mu(W))]^2\big]+\mathbb{E}_{P_{W}^{[n]}\otimes \mu}\big[(\ell(W,Z)-L_\mu(W))^2\big]-\nonumber\\&\ 2\mathbb{E}_{\mu}\big[\mathbb{E}_{W}^{\infty}[(\ell(W,Z)-L_\mu(W))]\mathbb{E}_{P_{W}^{[n]}}[(\ell(W,Z)-L_\mu(W))]\big]-\nonumber\\
            &\ \frac{1}{2}\mathbb{E}_{\mu}\big[\mathbb{E}_{W}^{\infty}[(\ell(W,Z)-L_\mu(W))^2]\big]+\frac{1}{2}\mathbb{E}_{\mu}\big[\mathbb{E}_{W}^{\infty}[(\ell(W,Z)-L_\mu(W))]^2\big]\bigg\}\nonumber\\
            =&\ \frac{\gamma^2}{2}\mathbb{E}_{\mu}\bigg[\mathbb{E}_{W}^{\infty}[(\ell(W,Z)-L_\mu(W))^2]-\mathbb{E}_{W}^{\infty}[(\ell(W,Z)-L_\mu(W))]^2\bigg]
        \end{align}
    where the final step uses the fact that $P_{W}^{[n]}\to P_{W}^{\infty}$. Comparing \eqref{equ:Thm5 nL} with \eqref{equ: n2Iskl} yields
    \begin{equation}
        \lim_{n\to\infty}n\cdot L(W;S)=\frac{1}{2}\lim_{n\to\infty}n\cdot I_{\mathrm{SKL} }(W;S).
    \end{equation}
    Eventually using $I_{\mathrm{SKL}}(W;S)=I(W;S)+L(W;S)$, we get
    \begin{equation}
        \lim_{n\to\infty}n\cdot I(W;S)=\lim_{n\to\infty}n\cdot L(W;S)=\frac{1}{2}\lim_{n\to\infty}n\cdot I_{\mathrm{SKL} }(W;S).
    \end{equation}
\end{proof}
\subsection{Proof of Theorem 5}
\begin{lemma}[Theorem 2 \cite{aminian2021exact}]\label{:Lem6}
    Suppose the non-negative loss function $\ell(w,Z)$ is $\sigma$-sub-Gaussian on the left-tail under distribution $\mu$ for all $w \in \mathcal{W}$. If we further assume $C_E\leq \frac{L(W;S)}{I(W;S)}$ for some $C_E\geq 0$, then for the $(\gamma,\pi(w),L_E(w,s))$-Gibbs algorithm, we have
    \begin{equation}
        0\leq {\rm gen}(P_{W|S}^\gamma,P_S)\leq \frac{2\sigma^2\gamma}{(1+C_E)n}
    \end{equation}
\end{lemma}
\begin{proof}
    Using Theorem \ref{:Thm5}, we get
    \begin{equation}
        \lim_{n\to\infty}\frac{L(W;S)}{I(W;S)}=1
    \end{equation}
    Therefore, for $\frac{\delta}{2}>0$, there exists an $N\in\mathbb{N}^+$ such that for all $n>N$,
    \begin{equation}
        \frac{L(W;S)}{I(W;S)}>1-\frac{\delta}{2}
    \end{equation}
    Using the fact that a bounded random variable $\ell\in[a,b]$ is $(b-a)/2$-sub-Gaussian, we apply Lemma \ref{:Lem6} and gets the result.
\end{proof}
\section{Mean Estimation Example}\label{appendixC}
\subsection{Preparation: SKL-divergence between Gaussians}
Suppose $X,Y$ follows a multi-Gaussian distribution: $[X,Y]\sim \mathcal{N}(m,\Sigma_{[X,Y]})$, where $\Sigma_{[X,Y]}$ can be written as:
\begin{equation}
\Sigma_{[X,Y]} = \left(
\begin{matrix}
\Sigma_X & \Sigma_{XY} \\
\Sigma_{YX} & \Sigma_Y
\end{matrix}
\right)
\end{equation}
The independent $\Tilde{X},\Tilde{Y}\sim P_XP_Y$ also follows a multi-gaussian distribution, that is $\Tilde{X},\Tilde{Y}\sim \mathcal{N}(m,\Sigma_{[\Tilde{X},\Tilde{Y}]})$, and:
\begin{equation}
\Sigma_{[\Tilde{X},\Tilde{Y}]} = \left(
\begin{matrix}
\Sigma_X & 0 \\
0 & \Sigma_Y
\end{matrix}
\right)
\end{equation}
then:
\begin{align}\label{equ:IsklXY}
I_{\mathrm{SKL} }(X;Y)&=\frac{1}{2}\log\frac{|\Sigma_X||\Sigma_Y|}{|\Sigma_{[X,Y]}|}+\frac{1}{2}\left(\log\frac{|\Sigma_{[X,Y]}|}{|\Sigma_X||\Sigma_Y|}+{\rm{tr}}(\Sigma_{[X,Y]}^{-1}\Sigma_{[\Tilde{X},\Tilde{Y}]}-I)\right)\\
&=\frac{1}{2}{\rm{tr}}(\Sigma_{[X,Y]}^{-1}\Sigma_{[\Tilde{X},\Tilde{Y}]}-I)
\end{align}\\
\subsection{Estimating the mean}
$S=\{Z_i\}_{i=1}^n$ is the training set, where $Z_i$ is a $d$ dimentional vector sampled i.i.d. from $\mathcal{N}(0_d,(\frac{1}{\sqrt{2\beta}})^2I_d)$. We consider the problem of learning the means of the dataset. For simplisity, we consider $d=1$. Loss function is the square error $\ell(w,Z)\triangleq \Vert w-Z\Vert_2$. We further choose our prior distribution $\pi(w)=\frac{1}{\sqrt{\pi}}\exp(-w^2)$

Under such settings, it can be calculated that
\begin{equation}
    P_{W,S}\sim \mathcal{N}(\{0\}^{n+1},\Sigma)
\end{equation}
where $\Sigma$ is an $n+1$ dimensional matrix as follow
\begin{equation}
    \Sigma=\frac{1}{2}
    \left(
    \begin{matrix}
    \frac{n\gamma\beta+n\beta+\gamma^2}{n(1+\gamma)^2\beta} & \frac{\gamma}{n\beta(1+\gamma)} & \frac{\gamma}{n\beta(1+\gamma)} &\frac{\gamma}{n\beta(1+\gamma)} & \cdots \\
    \frac{\gamma}{n\beta(1+\gamma)} & \frac{1}{\beta} & 0 & 0 & \cdots\\
    \frac{\gamma}{n\beta(1+\gamma)} & 0 & \frac{1}{\beta} & 0 & \cdots\\
    \frac{\gamma}{n\beta(1+\gamma)} & 0 & 0 & \frac{1}{\beta} & \cdots\\
    \vdots & \vdots & \vdots & \vdots & \ddots
    \end{matrix}
    \right)
\end{equation}
\begin{equation}
    \Sigma^{-1}=2\left(
    \begin{matrix}
    \gamma+1 & -\frac{\gamma}{n} & -\frac{\gamma}{n} &-\frac{\gamma}{n} & \cdots \\
    -\frac{\gamma}{n} & \frac{\gamma^2}{n^2(\gamma+1)}+\beta & \frac{\gamma^2}{n^2(\gamma+1)} & \frac{\gamma^2}{n^2(\gamma+1)} & \cdots\\
    -\frac{\gamma}{n} & \frac{\gamma^2}{n^2(\gamma+1)} & \frac{\gamma^2}{n^2(\gamma+1)}+\beta & \frac{\gamma^2}{n^2(\gamma+1)} & \cdots\\
    -\frac{\gamma}{n} & \frac{\gamma^2}{n^2(\gamma+1)} & \frac{\gamma^2}{n^2(\gamma+1)} & \frac{\gamma^2}{n^2(\gamma+1)}+\beta & \cdots\\
    \vdots & \vdots & \vdots & \vdots & \ddots
    \end{matrix}
    \right)
\end{equation}
With the joint distribution, we are able to calculate the symmetrized KL information of $S$ and $W$ using \eqref{equ:IsklXY}, which is the characterization of generalization error
    \begin{align}\label{equ:res_IsklWS_2}
    \gamma{\rm{gen}}(\mu,P_{W|S})&=I_{\mathrm{SKL} }(W;S)\nonumber\\
    &=\frac{1}{2}{\rm{tr}}(\Sigma_{[W,S]}^{-1}\Sigma_{[\Tilde{W},\Tilde{S}]}-I)\nonumber\\
    &=\frac{\gamma^2}{n\beta(1+\gamma)}
    \end{align}
Next, we consider the symmetrized KL information of $Z_i$ and $W$. Similarly, we get
    \begin{align}\label{equ:res_IsklWZ_2}
    I_{\mathrm{SKL} }(W;Z_i)&=\frac{1}{2}{\rm{tr}}(\Sigma_{[W,Z_i]}^{-1}\Sigma_{[\Tilde{W},\Tilde{Z_i}]}-I)\nonumber\\
    &=\frac{1}{2}{\rm{tr}}\left(
    \left(
    \begin{matrix}
    \frac{n\gamma\beta+n\beta+\gamma^2}{n(1+\gamma)^2\beta} & \frac{\gamma}{n\beta(1+\gamma)}\\
    \frac{\gamma}{n\beta(1+\gamma)} & \frac{1}{\beta}
    \end{matrix}
    \right)^{-1}\cdot
    \left(
    \begin{matrix}
    \frac{n\gamma\beta+n\beta+\gamma^2}{n(1+\gamma)^2\beta} & 0\\
    0 & \frac{1}{\beta}
    \end{matrix}
    \right)-I
    \right)\nonumber\\
    &=\frac{\gamma^2}{n^2\beta(1+\gamma)+\gamma^2(n-1)}
    \end{align}
From the result of \eqref{equ:res_IsklWS_2} and \eqref{equ:res_IsklWZ_2}, we get
\begin{equation}\label{equ:gap_example_2}
    \begin{aligned}
    \sum_{i=1}^nI_{\mathrm{SKL} }(W;Z_i)-I_{\mathrm{SKL} }(W;S) = \Theta\bigg(\frac{1}{n^2}\bigg)=o(I_{\mathrm{SKL} }(W;S))
    \end{aligned}
\end{equation}
We can see that \eqref{equ:res_IsklWZ_2} and \eqref{equ:gap_example_2} corresponds to Theorem \ref{:Thm2} and Theorem \ref{:Thm4}, despite we are not considering a bounded loss function.

We further investigate the Jensen gap term $J_i(w,z_i)$, since as stated previously, the key to prove Theorem \ref{:Thm4} is that the effect of variable $z_i$ is order-wise smaller than that of $w$. In estimating the mean problem, we can calculate that
    \begin{align}
    J_i^{[n]}(w,z_i)
    =&\log\left(\sqrt\frac{n^2(1+\gamma)^2\beta}{\pi (n^2\beta(1+\gamma)+(n-1)\gamma^2)}\right)-\log\left(\sqrt{\frac{1+\gamma}{\pi}}\right)+\frac{n-1}{n^2}\cdot \frac{\gamma^2}{1+\gamma}\cdot\frac{1}{2\beta}\nonumber\\
    &+w^2(1+\gamma)\left[1-\frac{n^2(1+\gamma)^2\beta}{n^2\beta(1+\gamma)+(n-1)\gamma^2}\right]-\frac{2\gamma}{n}wz_i\left[1-\frac{n^2(1+\gamma)^2\beta}{n^2\beta(1+\gamma)+(n-1)\gamma^2}\right]\nonumber\\
    &+\frac{\gamma^2}{n^2(1+\gamma)}z_i^2\left[1-\frac{n^2(1+\gamma)^2\beta}{n^2\beta(1+\gamma)+(n-1)\gamma^2}\right]\\
    =& w^2\Theta\bigg(\frac{1}{n}\bigg)+wz_1\Theta\bigg(\frac{1}{n^2}\bigg)+z_1^2\Theta\bigg(\frac{1}{n^3}\bigg)+\Theta\bigg(\frac{1}{n^2}\bigg).
    \end{align}
We can see that the contribution of $z_i$ term is $\Theta\big(\frac{1}{n^2}\big)$ which is indeed 1 order smaller than those terms not containing $z_i$.

Finally, we provide similar analysis on mutual information.
    \begin{align}
        I(W;Z_i)&=\frac{1}{2}\log\bigg(1+\frac{\gamma^2}{n^2(1+\gamma)\beta+(n-1)\gamma^2}\bigg)\\
        &=\Theta\bigg(\frac{1}{n^2}\bigg),
    \end{align}
and
    \begin{align}
        I(W;S)&=\frac{1}{2}\log\bigg(1+\frac{\gamma^2}{n\beta(1+\gamma)}\bigg)\\
        &\sim \frac{1}{2}\cdot\frac{\gamma^2}{n\beta(1+\gamma)}\\
        &=\frac{1}{2}I_{\mathrm{SKL} }(W;S),
    \end{align}
which corresponds to our result in Theorem \ref{:Thm5}.

\end{document}